\documentclass[a4paper,UKenglish,cleveref,autoref,thm-restate]{lipics-v2021}

\usepackage{mathtools, xparse}
\usepackage{graphicx}
\usepackage{epstopdf}
\usepackage{bm}
\usepackage{scalerel}
\usepackage{amsfonts,amsmath}
\usepackage{algorithm,algorithmic}
\usepackage{mathtools}
\usepackage{multirow}
\usepackage{verbatim}
\usepackage{url}
\usepackage{comment}
\usepackage{mathtools}
\usepackage{epsf,pgf,graphicx}
\usepackage{color}
\usepackage{tikz,pgf}
\usepackage{diagbox}
\usepackage{makecell}
\usepackage{color}
\usepackage{braket}
\usepackage{mathtools}
\usepackage{float}
\usepackage{tikz}
\usetikzlibrary {positioning}
\usepackage{pgf}
\usepackage{qcircuit}

\usepackage{multirow}
\usepackage{multicol}
\usepackage{listings}
\usepackage{array}
\usepackage{makecell}

\newcommand{\op}{\oplus}

\DeclarePairedDelimiter{\abs}{\lvert}{\rvert}
\newcommand{\pdeg}{{\sf pdeg}}

\newcommand{\pr}{\mathcal{P}}
\newcommand{\xv}{\mathbf{x}}
\newcommand{\yv}{\mathbf{y}}

\newcommand{\lin}{\mathcal{L}}
\newcommand{\x}{\bm{\hat{x}}}

\newcommand{\var}[1]{\mathbf{#1}}

\newcommand{\sym}[2]{\widehat{sym}^{#1}_{#2}}

\newcommand{\F}{\mathbb{F}_2}
\newcommand{\lf}{\bf{L_{n_2}}}
\newcommand{\M}{\mathbb{MM}}

\newtheorem{construction}{Construction}
\newtheorem{result}{Result}

\title{
On Boolean Functions with Low Polynomial Degree and 
Higher Order Sensitivity
}
\titlerunning{Low Polynomial Degree and Higher Order Sensitivity}
\author{Subhamoy Maitra}
{Applied Statistics Unit, Indian Statistical Institute, Kolkata, India}
{subho@isical.ac.in}
{}
{}
\author
{Chandra Sekhar Mukherjee}
{Indian Statistical Institute, Kolkata, India}
{chandrasekhar.mukherjee07@gmail.com}
{}
{}
\author{Pantelimon Stanica}
{Naval Postgraduate School, Monterey, USA}
{pstanica@nps.edu}
{}
{}
\author{Deng Tang}
{Southwest Jiaotong University, Chengdu, China}
{dtang@foxmail.com}
{}
{}
\authorrunning{Maitra, Mukherjee, Stanica and Tang}

\begin{document}
\nolinenumbers

\maketitle
\keywords{Higher Order Sensitivity, Resiliency, Maiorana-McFarland Construction, Polynomial Degree, Sensitivity, Separation.}
\begin{abstract}
Boolean functions are important primitives in different domains of cryptology, complexity and coding theory. 
In this paper, we connect the tools from cryptology and complexity theory in the domain of 
Boolean functions with low polynomial degree and high sensitivity.
It is well known that the polynomial degree of of a Boolean function and its resiliency are directly connected. 
Using this connection we analyze the polynomial degree-sensitivity values through the lens of resiliency, demonstrating 
existence and non-existence results of functions with low polynomial 
degree and high sensitivity on small number of variables (upto 10).
In this process, borrowing an idea from complexity theory, we show that one can implement resilient Boolean functions on a large number of 
variables with linear size and logarithmic depth. 
Finally, we extend the notion of sensitivity to higher order and note that the existing construction idea of Nisan and Szegedy (1994) can provide 
only constant higher order sensitivity when aiming for polynomial degree of $n-\omega(1)$. 
In this direction, we present a construction with low ($n-\omega(1)$) polynomial degree and 
super-constant $\omega(1)$ order sensitivity exploiting Maiorana-McFarland constructions, 
that we borrow from construction of resilient functions.
The questions we raise identify novel combinatorial problems in the domain of Boolean functions.
\end{abstract}

\section{Introduction \label{sec:1} }
Real polynomial degree ($\pdeg(f)$) and sensitivity ($s(f)$) are two central properties of Boolean functions in complexity theory, and have 
been studied extensively over the past three decades. These notions have important implications in the domain of query complexity~\cite{BW02} 
(and not only), where finding functions with lower polynomial degree than sensitivity generates more candidates for obtaining super-linear 
separation between two of the query complexity models, the classical deterministic and exact quantum models. A detailed study of these properties 
and the relations can be found in~\cite{BW02}.

Determining the maximum possible separation between sensitivity and polynomial degree 
is an open problem that has been studied widely. Implicitly, the problem reduces to finding the separation between the number of variables 
and the real polynomial degree in a fully sensitive function (a function $f$ on $n$ variables with $s(f)=n$). 
Informally, the sufficient and necessary condition for obtaining full sensitivity in a function $f$ on $n$ variables is to have 
an input point $\xv \in \{0,1\}^n$ such that $f(\xv)=\overline{f(\xv^i)}, 1 \leq i \leq n$, where $\xv^i$ is obtained 
by altering the value of the $i$-th bit of $\xv$. Thus, we fix the output corresponding to some $n+1$ input points, of the 
total $2^n$ input points. Interestingly, this greatly restricts the real polynomial degree of the function. 
Without any restriction, a function that depends on all of its variables can have $\pdeg(f)$ 
as low as $\mathcal{O}(\log n)$. However, one of the seminal papers~\cite{NW95} in the study of 
Boolean functions dictate that $\pdeg(f)= \Omega\left({s(f)}^{\frac{1}{2}} \right)$.
Furthermore, apart from the famous recursive amplification method (which is also known as the function composition), 
there does not exist any known method to obtain functions with non-constant separation 
between $s(f)$ and $\pdeg(f)$. In fact, the maximum separation 
known between $s(f)$ and $\pdeg(f)$ is achieved by finding 
a $6$ variable function with $s(f)=6$ and $\pdeg(f)=3$ (due to Kushilevitz~\cite{NS94})
and then recursively amplifying it. This results in a function $f$ on $n=6^d$ 
variables with full sensitivity ($n$)
and polynomial degree of $3^d$ so that $s(f)={\pdeg(f)}^{\log_36} \approx {\pdeg(f)}^{1.63}$~\cite{NW95}.
On the other hand, the real polynomial degree is intrinsically connected to the cryptographically important 
property of resiliency. Specifically if a function $f$ on $n$ variables has $\pdeg(f)=m$
then the function $g= f \op \lin_n$ is $(n-m-1)$-resilient, where $\lin_n$ is the all variable (symmetric) linear function. 
In this paper we refer to $g$ and $f$ as each other's dual. In this regard, we define the  dual sensitivity ($ds(f)$) property,
where a function $g$ has full dual sensitivity if and only if there exists a point 
$\xv \in \{0,1\}^n$ such that $g(\xv)=g(\xv^i), 1 \leq i \leq n$.
This results in a one-to-one connection between the $\pdeg(f)$-$s(f)$
relationship and resiliency order-$ds(g)$ relationship, where $g$ is the dual of $f$. We should remark here that the resilient Boolean functions 
have received a lot of interest in construction of symmetric ciphers as evident from~\cite{KM06,MS99,SM02}.

In this paper, we show that the techniques from complexity theory and cryptology can supplement each other with respect to combinatorial aspects of Boolean functions which are important in their own interests,  and extend the notion of sensitivity to higher order towards a better understanding how fixing of outputs with respect to flipping of more than one input bits can effect the lower bound on the polynomial degree of a function. 
That is, we use different properties of resiliency and polynomial degree to obtain 
results and constructions that apply to both paradigms of cryptology and complexity theory.

\subsection{Contribution and Organization}
Our paper is centered around two related concepts. In the first part we exploit the well known connection between resiliency and polynomial degree.

We observe an one to one connection between ``low polynomial degree-high sensitivity'' and ``high resiliency-dual sensitivity'' of Boolean functions, via their Fourier spectrum based definitions. We use this connection to search for separation between $s(f)$ and $\pdeg(f)$ 
for functions on all small number of variables through the resiliency approach.
We find new classes of functions with maximum $s(f)-pdeg(f)$ separation and obtain super-linear separation between $n$ and $\pdeg(f)$ for second and third order sensitive functions. Specifically we find the following which were not known earlier.
\begin{itemize}
\item There exists second order six variable functions with $\pdeg(f)=3$. That is the maximum known separation for $n$ and $\pdeg$ is same for first and second order sensitivity.
\item There does not exist any $7$ variable fully sensitive function with $\pdeg(f)=3$.
\end{itemize}
Further, We analyze the famous recursive amplification (RA) method. 
We show that its generalization allow us to obtain additional classes of functions with super-linear separation between $s(f)$ and $\pdeg(f)$. 
Next, using the RA method we design efficient circuits 
(linear size and logarithmic depth (in $n$) for highly resilient 
functions $(n-o(n))$, improving upon best known results. 
We further obtain bounds on cryptographic properties of 
these functions, such as nonlinearity and describe different trade-offs. 

We use our resiliency based search method to obtain second  order sensitive $6$ variable functions with  $\pdeg(f)=3$. Coupled with the modified recursive amplification method that we propose in Section~\ref{sec:3} (Further elaborated in Theorem~\ref{th:ranew} in Appendix~\ref{app:sec:3} ), this gives us second order functions $f^u$ with $n=6^u$ variables and $\pdeg(f^u)=n^{\frac{\log 3}{\log 6}}$, matching the best known bound between $n$ and $\pdeg(f)$ for first order sensitivity. This raises the question of whether asymptotic separation between $n$ and $\pdeg(f)$ is a strictly decreasing 
function when plotted against sensitivity order.

In the second part, we take a deeper look on higher ($k$-th) order sensitivity. 
Sensitivity of a function $f: \F^n \rightarrow \F$, at a point $\xv \in \F^n$,  is the number of bits of $\xv$ so that, upon flipping any one of them, the output of the function also flips. Extending this, a function is called $k$-th order sensitive, if there is a point $\xv$ such that changing any $1 \leq i \leq k$ of the input bits changes the output of the function. This is a much stronger notion than simple sensitivity, where a fully sensitive function of present literature ($s(f)=n$) is a first order one. We are interested in understanding the level of restriction higher order sensitivity (dual sensitivity) has on polynomial degree (resiliency).
\begin{itemize}
\item In this direction we observe for constant sensitivity order, we can design functions with sub-linear polynomial degree $o(n)$ using the recursive amplification method. However, even for functions with $\pdeg(f)= \mathcal{O}\left(n-o(n) \right)$,
we can only have constant sensitivity order through this method.  
\item Finally, we show that the Maiorana-McFarland (MM) construction can be used to show super-constant separation between $n$ and $\pdeg(f)$ for functions with super-constant ($\omega(1)$) sensitivity order. We first design a first order sensitive function with $\pdeg(f)=n \log n$, which is worse than the recursive amplification method as mentioned above. However, by concatenating non-linear symmetric function in one of the half spaces of the MM function, we obtain the desirable trade-off. We obtain $k$-th order sensitive functions, with $\pdeg(f)=n- \frac{\log \left( \frac{n}{2}-k \right) - \log (k)}{k}$.
This gives us functions with $s(f)=\omega(1)$ and $\pdeg(f)=n-\omega(1)$. For example, 
if we set $k= \log \log (n)$, we get $\log \log (n)$-order sensitive functions with polynomial 
degree of $\approx n- \frac{\log n}{\log \log n}$.
\end{itemize}
The paper is organized as follows. In Section~\ref{sec:1-3} we recall the definitions of sensitivity and polynomial degree and the connection of polynomial degree with resiliency. Then we briefly describe the concept dual sensitivity, followed 
by that of higher order sensitivity and higher order dual sensitivity with the fundamental equivalences. 
In Section~\ref{sec:2} we obtain the search based results on sensitivity and higher order sensitivity vs. polynomial degree.
In Section~\ref{sec:3} we study the recursive amplification method and obtain the related results. 
Section~\ref{sec:4} is dedicated towards finding functions with higher order sensitivity and lower than $n$ ($n-w(1)$) polynomial degree.
We conclude the paper in Section~\ref{sec:5}.

\subsection{Preliminaries \label{sec:1-3}}
The definitions of resiliency and polynomial degree are based on the Fourier spectrum of a Boolean function~\cite{bool}.
A function $f$ has polynomial degree $k$ iff its Fourier spectrum values $W_f(\xv)$ are $0$ for all $\xv \in \{0,1\}^n$ with hamming weight  $wt$ of $\xv$ being greater than $k$.
A function $f$ is $k$-resilient if and only if 
its Fourier spectrum values are $0$ for all $\xv: wt(\xv) \leq k$.
We call a function to be $k$-th order resilient when it is $k$-resilient but not $k+1$-resilient.
Given a function $f$,  $W_f(\xv)= W_{f \op \lin_n}(\bar{\xv})$
where $\lin_n$ is the linear function on $n$ variables with $\bar{\xv}$
obtained by flipping each bit of $\x \in \F^n$. 
These structural arguments gave rise to the famous result 
connecting the resiliency order and polynomial degree of Boolean functions. 
\begin{theorem}[\cite{bool}, page 150]
\label{th:old}
If a function $g$ is $k$-th order resilient then
the function $f= g \op \lin_n$ will have a polynomial degree equal to $n-k-1$, where $\lin_n= \oplus_{i=1}^n x_i$. 
\end{theorem}

Next, we study the notion of sensitivity and also, that 
of dual sensitivity, which we define to better understand the connection between
resiliency and polynomial degree.

\subsection{Sensitivity}
Sensitivity $s(f)$ is one of the most studied properties of Boolean function.
For any $\xv \in \F^n$, we let  $\xv^i$ to be $\xv$  
with the $i$-th bit of $\xv$ flipped (complemented).
The sensitivity of a Boolean function $f: \F^n \rightarrow \F$ at a point $\xv$
can be defined as $s(f,\xv)=|i \in [n]: f(\xv) \neq f(\xv^i)|$,
and the sensitivity of a function is 
$
s(f)=\max_{\xv \in \F^n} s(f,\xv).
$
It is natural to consider the situation where we want the function to have the same 
value even if multiple input bits of $\xv$ are flipped regardless of their position. 
In this direction, we define the $k$-th order sensitivity of a Boolean function. 

{\Large \bf $k$-th order sensitivity: }
For  any set $S \subseteq [n]$ and the input point $\xv \in \F^n$ we define
$\xv^{(S)}$ as the input point obtained by flipping the $j$-th bit of $\xv$ for all $j \in S$.
Sensitivity is defined around the notion of flipping any single component corresponding to 
a given input where the output of the function remains unchanged. 
In this regard we define $k$-th order sensitivity of a function in the following manner.

\begin{definition}
We call a function $f: \F^n \rightarrow \F$ $k$-th order sensitive if there exists
$\xv \in \F^n$ such that 
$
f(\xv) \neq f \left( \xv^{(S)} \right),~ \forall S \subseteq [n],\ 1 \leq \abs{S} \leq k.  
$
\end{definition}

That is, $f$ is $k$-th order sensitive if there exists an input so that flipping any $i \leq k$ of 
the component bits of the input, changes the function's output.
Thus a first order sensitive function is simply a function with $s(f)=n$. 
The main implication of $k$-th order sensitivity is that it indeed further restricts how 
low the degree of the real polynomial corresponding to the function can be. 
Without any restrictions we know $\pdeg(f)$ can be as low as $\log n$ for functions that depend on $n$ variables. If we fix $s(f)=n$ then the polynomial degree is $\Omega \left( \sqrt{n} \right)$. The paper is centered around obtaining functions with $\omega(1)$-order sensitivity and $n-\omega(1)$-polynomial degree.

\subsection{Dual sensitivity}
Given a function $f$ on $n$ variables with polynomial degree $m$, 
we call the function $g=f \op \lin_n$, the dual of $f$, which has $n-m-1$ resiliency. 
The dual sensitivity of a function $f: \F^n \rightarrow \F$ at a point 
$\xv$ is defined as $ds(f, \xv) = \abs{i \in [n] : f(\xv) = f(\xv^i)}$.
The dual sensitivity of $f$  is $ds(f)=\max_{ \xv \in \F^n} ds(f, \xv).$
This notion can be extended to $k$-th order dual sensitivity in the following manner.

\begin{definition}
We say a function $f$ is $k$-th order dual sensitive if there exists $\xv \in \F^n$,
such that, for all $j: 1 \leq j \leq k$ we have:
\begin{itemize}
\item If $j \equiv 0 \bmod 2$, then $f(\xv) \neq f \left(\xv^{(S)} \right),\ \forall S \subseteq [n] \text{ with } \abs{S}=j$.
\item If $j \equiv 1 \bmod 2$, then $f(\xv) = f \left(\xv^{(S)} \right),\ \forall S \subseteq [n] \text{ with } \abs{S}=j$.
\end{itemize}
\end{definition}

That is, a function is $k$-th order dual sensitive if there is an input point such that if we flip the values of 
any odd number $\le k$ of input bits then the function's output remains unchanged and if we flip any even number $\le k$  of input bits
then the function's output gets complemented. 

\begin{proposition}
\label{prop:ds}
A function $f$ on $n$ variables is $k$-th order sensitive if and only if its dual $g=f \op \lin_n$ is  $k$-th order dual sensitive.  
\end{proposition}

\begin{note}
We use the following notations:
\begin{itemize}
\item An $(n,k,p)$-function is a Boolean function of $n$ variables that is $k$-th order sensitive and  has real polynomial degree at most $p$.

\item An $[n,k,m]$-function is a Boolean function of $n$ variables that is $k$-th order dual sensitive and  $m$-resilient. 
\end{itemize}
\end{note}
Thus we have the following proposition.
\begin{proposition}
\label{prop:tf}
Thus, if $f$ is an $(n,k,p)$-function then $g= f \op \lin_n$ is an $[n,k,n-p-1]$-function.
\end{proposition}
Let us now move onto the search based results.

\section{Search On Small Variables \label{sec:2}}
As we shall observe in Section~\ref{sec:3}, 
upon some modification, the recursive amplification method can be used to obtain $k$-th order sensitive functions $f^u$ on $d^u$ variables
with polynomial degree of $p^u$, starting from a function $f$ on $d$ variables and $\pdeg(f)=p$. Here $p$ is called the base function. 
Thus results of low $\pdeg$ of $k$-th order sensitive functions on small variables directly generate super-linear separations between $n$ and $\pdeg$ for $k$-th order sensitive functions. 

For example, if we obtain a fully sensitive (first oder sensitive) function on $7$ variables with $\pdeg(f)=3$, or a $10$ variable first order sensitive function with polynomial degree $4$, then it would improve upon the best known separation between $s(f)$ and $\pdeg(f)$. 
In this direction, 
the functions on up to $5$ variables can be exhaustively searched to obtain all existing combinations. However, for functions on $6$ and more variables, an exhaustive search is not possible given
the size of search space ($2^{64}$ for $n=6$, $2^{128}$ for $n=7$ and so on), and we instead use the properties of resiliency and dual sensitivity to completely exhaust the case of fully sensitive functions for $n=6$ and $n=7$ in terms of obtaining all functions and proving non existence respectively.

\subsection{$[4,-,0]$-functions:}
It can be checked with a simple search that there does not exist any fully sensitive function on $4$ variables with $\pdeg(f)=2$. In fact, if that would have been the case then we could 
use the recursive amplification method to obtain a function $f^u$ on $4^u$ variables with $s(f^u)=4^u$ and $\pdeg(f^u)=2^u$, which would give us quadratic separation between sensitivity
and polynomial degree, demonstrating a tight lower bound.
Thus we look into the $(4,-,3)$ functions, which are duals of $[4,-,0]$ functions. 
We have the following counts.
\begin{itemize}

\item There are approximately $ 3760 (\approx 2^{11.87})$ many  $[4,1,0]$-functions. 
\item Only  $256(=2^8)$ of these functions are $[4,2,0]$-functions and there are no $[4,3,0]$-functions.

\end{itemize}

\subsection{$[5,-,1]$-functions:}
In the case of $5$ variable functions, the possible polynomial degree   is between $3$ and $4$.
As we have already observed the case of $0$-resiliency ($n-1$ polynomial degree in the dual) 
for $3$ and $4$, we compute the $1$-resilient functions in this case, and get the following counts:
\begin{align*}
& \#\, [5,1,1]\text{-functions}=12304 (\approx 2^{13.58} ), &\#\, [5,3,1]\text{-functions}&=0\\
& \#\, [5,2,1]\text{-functions}=2464  (\approx 2^{11.2} ), & \#\,[5,4,1]\text{-functions}&=0.
\end{align*}

Finally let us look into the case of $6$ variable functions, for which we have the best base function
for first order sensitivity, which is the Kushilevitz function.  

\subsection{$[6,-,2]$-functions:}
\label{sec:4-6}

There are total $2^{64}$ Boolean functions on $6$ variables, and checking the resiliency
and sensitivity of all possible functions requires computational resources that is unattainable. 
We instead use properties of dual sensitivity and resiliency to obtain all possible $[6,1,2]$-functions
by concatenating the truth tables of two $5$ variable functions.  
Any $6$ variable function $f$ can be written as $f(x_1,\ldots ,x_6)=(1 \op x_6)f_2(x_1,\ldots ,x_5) \op x_6f_2(x_1,\ldots ,x_5)$
Where $f_1$ and $f_2$ are functions on $5$ variables. Then we have the following constraints on the properties of $f$.

\begin{enumerate}

\item If $f$ is $2$-resilient then either both $f_1$ and $f_2$ are $1$-resilient or both are $2$-resilient~\cite{MS99}. 

\item If $f$ is fully dual sensitive then at least one of $f_1$ and $f_2$ are fully dual sensitive. 
This is easy to see as if neither $f_1$ nor $f_2$ are dual sensitive then there is no input point 
for which the whole function can have full dual sensitivity. 

\end{enumerate}

Now we have only $2^{13}$ many $[5,1,1]$-functions and $2^{18}$ many $[5,0,1]$-functions, which reduces
the effective search space to approximately $2^{33}$ from the naive $2^{64}$. 
Using these constraints we get the full characterization of $[6,-,2]$-functions, which was not previously reported.

\begin{itemize}

\item We find that there are $33632 (\approx 2^{15.03} )$ many $[6,1,2]$-functions.
Here it should be noted that the dual of any such function is a $(6,1,3)$-function. 
We can use the modified recursive amplification technique of Theorem~\ref{th:ranew}
on all such functions to obtain $(6^u,1,3^u)$-functions,
which gives us the best known separation between sensitivity and polynomial degree, 
same as the function by Kushilevitz~\cite{NS94}.

\item We also get $192  (\approx 2^{7.6} )$ many $[6,2,2]$-functions, and this gives us the
maximum super-linear separation between number of variables and real polynomial degree 
in second order sensitive functions, which is $\pdeg(f)=n^{\frac{\log 3}{\log 6}}$,
which is also the currently best known separation for first order sensitivity.

Furthermore, there is no $[6,>2,2]$-functions.
\end{itemize}

\subsection{Nonexistence of $(7,1,3)$-functions and Searching the Rotation Symmetric Functions}
The existence or non-existence of a $(7,1,3)$-function is central to understanding the maximum 
separation between $s(f)$ and $\pdeg(f)$. If there does exist a $(7,1,3)$-function then we can
obtain a $(7^u,1,3^u)$-function using the recursive amplification method, which gives $s(f)={\pdeg(f)}^{\frac{\log 7}{\log 3}}$, improving on the best known result. 
However the total number of functions on $7$ variables is 
$2^{128}$ and therefore checking all functions for this 
profile through brute force is not computationally possible. 
Against this background we use a mixed integer linear program (MILP) to investigate the existence of  such a function. If $f$ is a $7$ variable Boolean function with $s(f)=7$ and $\pdeg(f)=3$ iff
there is a vector ${\bf x}\in\mathbb{F}_2^7$ such that $f({\bf x})\oplus ({\bf x}^i)=1$ for every $1\leq i \leq 7$ and $W_f({\bf u})=0$ for every ${\bf u}\in\mathbb{F}_2^7$ with $wt({\bf u} \geq 4$. For every ${\bf x}\in \mathbb{F}_2^7$. We run the MILP and it returns no solution for all choice of $\xv$. This shows there are no $(7,1,3)$-functions. One can refer to Appendix~\ref{app:713} for a formal description of the constraints.

Even with our strategy, it not possible to search for all fully sensitive and higher order sensitive functions on more than $7$ variables 
because of the  size of the search space. In this regard we search $8,9$ and $10$ variable rotation symmetric functions, which is another 
cryptographically important class of functions to obtain with fully sensitive (first order sensitive functions) using least possible 
polynomial degree (maximum resiliency in the dual function). Our findings can be found in 
Appendix~\ref{app:rot}. Let us now proceed towards the recursive amplification method. 

\section{The Recursive Amplification Method \label{sec:3}}
We have noted that fixing the value of a function corresponding to $n+1$ input points to make a function fully sensitive, will restrict the polynomial degree to $\Omega({\sqrt{n}})$.
The best known results in this paradigm is derived through the recursive amplification method,
which is also the function composition method. This is a well known technique that is used to obtain super-linear separation between $n$ and $\pdeg(f)$ and
is also used to obtain super-linear separation between $s(f)$ and $\pdeg(f)$. 
In this section we use this technique and obtain the following results:
\begin{enumerate}
\item A slight modification of the recursive amplification method to obtain super-linear separation between $s(f)$ and $\pdeg(f)$ by starting from any candidate base function. 

\item We build highly resilient functions with good nonlinearity, $\mathcal{O}(n)$ circuit size and $\mathcal{O}(\log n)$ circuit depth.

\item We obtain super-linear separation between number of variables($n$) and
polynomial degree ($\pdeg(f)$) for functions with constant order sensitivity.
\end{enumerate}

Recursive amplification was used to obtain the largest known separation
between sensitivity and polynomial degree of Boolean functions~\cite{BW02}, as well as the first example of 
separation between exact quantum query complexity and deterministic query complexity~\cite{AMB13}, among other 
separation results. 
Let $f$ be a function on $d$ variables $x_1, x_2 \ldots x_d$ with polynomial degree $p$. 
Then the recursive amplification method generates the function $f^u$ on $d^u$ variables as:
\begin{enumerate}
\item $f^1=f$.
\item $f^{i+1}\left(x_1, \ldots ,x_{d^{i+1}} \right)= f\left( f^i\left(x_1, \ldots, x_{d^i} \right), \ldots ,
f^i \left(x_{(d-1)d^i+1}, \ldots , x_{d^{i+1}} \right) \right)$. 
\end{enumerate}
Then for any $f$ we have $\pdeg(f^u)=p^u$.
Thus if the sensitivity also gets amplified, we could start with any $d$ variable function with and obtain $f^u$ with super-linear $s(f^u)-\pdeg(f^u)$ whenever $s(f)>\pdeg(f)$. 
However, sensitivity is not always amplified in the similar manner,
and $s(f^u)$ can be arbitrarily low. To this end we propose a construction so that we can get super-linear separation between $s()$ and $\pdeg()$ starting from any function. Furthermore the results also follow for higher-order sensitivity. 
Let ${\yv}^{i} \in \F^{d^i}$ be obtained by concatenating $d$ copies of  ${\yv}^{i-1}$.
Then we define the amplification method w.r.t a base function $f$ on $d$ variables as
$
f^i= f\Big( f^{i-1}(x_1,\ldots ,x_{d^{i-1}}) \op f^{i-1}(\yv^{i-1}) \op y_1
,\ldots, 
f^{i-1}(x_{jd^{i-1}+1}, \ldots ,x_{(j+1)d^{i-1}} ) \op f^{i-1}(\yv^{i-1}) \op y_j
, \ldots, 
f^{i-1}(x_{(d-1)d^{i-1}+1}, \ldots, x_{d^i}) \op f^{i-1}(\yv^{i-1}) \op y_d \Big).
$.
Then $f^u$ is a $k$-th order sensitive function on $n=d^u$ variables and 
${\pdeg(f^u)=\pdeg(f)}^u$,
with $k$-th order sensitivity achieved at the input point ${\yv}^u$.
One can refer to Theorem~\ref{th:ranew} in Appendix~\ref{app:sec:3} 
for the formal representation and proof. 

Now we look into the functions of $4$ variables and then discuss how the recursive amplification 
method can be used to obtain highly resilient functions with good nonlinearity. 

\subsection{Low cost resilient functions with recursive amplification}
Let us consider a function $f$ on $d$ variables with $\pdeg(f)=p<d$, where $d$ is a constant.
Then we can recursively amplify the function to obtain a function  $f^u$ on $n=d^u$ variables 
with $\pdeg(f^u)=p^u=n^{\frac{\log p}{\log d}}$. 
Now if we add the all variables linear function to it we get an $n$ variable function with 
$n-n^{\frac{\log p}{\log d}}-1$ resiliency. However, there already exists many methods of obtaining 
Boolean functions with high resiliency and other cryptographically important properties such 
as high nonlinearity. 

Here the advantage of the recursive amplification method is the circuit size for building such functions. 
Building efficient low depth circuits for 
cryptographically important functions with large number of input variables is a challenging problem. 
In this regard the work by Sarkar et al.~\cite[2003]{SM03} is important. This work shows how to start with an
$m$-resilient function on some $d$ variables and generate an $m+u$-resilient function on $n=d+u$ variables 
that requires $\mathcal{O}(u)$ depth, which is effectively $\mathcal{O}(n)$ as $d$ is constant for any given 
construction. In fact, this has been the best known result in this direction for almost two decades in 
building efficient circuits for resilient functions on large variables starting from base functions. Improving on this, we have the following result.
\begin{result}
\label{th:recres:1}
Given a function $f$ on $d$ variables with $\pdeg(f)=p<d$ we can obtain a function 
on $g^u$ on $n=d^u$ variables with resiliency $n-n^{\frac{\log p}{\log d}}-1$ such that 
there is a circuit of linear size and logarithmic depth (in $n$) for it.
Here $g^u$ is the dual of $f^u$ where $f^u$ is the function on $n^u$ variables obtained by recursively amplifying $f$.
\end{result}
The proof can be found in Theorem~\ref{th:recres} in the appendix,
followed by Figure~\ref{fig:rec} that gives an example of building a $9$ variable function using $4$ instances of the circuit $C_f$ corresponding to a $3$-variable function $f$.
We refer to Appendix~\ref{app:crypto-prop} for elaborate discussion on the nonlinearity lower bounds we have derived for these highly resilient functions, along with algebraic degree-resiliency trade-offs. 
Finally, we show that we can obtain super-linear separation between $n$ and $\pdeg(f)$ for functions with any constant order sensitivity. This raises the interesting problem of understanding the nature maximum super-linear separation possible between $n$ and $\pdeg(f)$ with increasing, constant order sensitivity $k$. Specifically we have the following result. 

\begin{result}
Given any constant $k$ there exists a $k$-th order sensitive function $f$ on $n$ variables such that 
$\pdeg(f)= n^{\frac{\log k}{\log k+1}}$, if $k$ is even and 
$\pdeg(f)= n^{\frac{\log k+1}{\log k+2}}$, if $k$ is odd.
\end{result}
One may refer to Section~\ref{app:const-amp} for the detailed formal explanation. 

\section{Higher Order Sensitivity \label{sec:4}}
Until now we have discussed functions with a constant higher order sensitivity, and have found classes of functions $f$ defined on the number of variables $n$ for which we could obtain super-linear separation 
between $n$ and $\pdeg(f)$ using the recursive amplification method. However, we cannot obtain any $t(n)$-order sensitive (or dual sensitive) function, where $t(n)$ is an increasing function on $n$ using any recursive amplification process, whenever
we intend the function to have less than $n$ polynomial degree. 

\begin{theorem}
\label{th:sc}
The general recursive amplification process cannot obtain a function that has super-constant
order of sensitivity where the polynomial degree of less than the number of variables, 
where the recursive amplification process is defined as 
\begin{itemize}

\item A base function $f^1$ on some $d$ variables.

\item $f^k=f^1\Big(\hat{f}^{k-1}, \ldots ,\hat{f}^{k-1}  \Big)$ 
where $\hat{f}^{k-1} \in \{f^{k-1}, \overline{f^{k-1}}  \}$.
\end{itemize}
\end{theorem} 
\begin{proof}
Let us consider any function $f$ on $d$ variables and $t(n)$-order sensitivity where $t(n)=\Omega(1)$. 
Then $f^n$ is a function on $d^n$ variables and there exists $n_0 \in \mathbb{N}$ such that
$t(d^{n_0})>d$. However it is easy to see (via induction) that any function built with a base function on $d$ variables and sensitivity order less than $d$ cannot be $d$-th order sensitive.
\end{proof}

Thus the recursive amplification process does not help us anymore when we consider 
$\omega(1)$-order sensitive functions. In this regard we next explore the class of 
Maiorana-McFarland (MM)  constructions, a heavily studied class for cryptographic and coding theoretic purposes.

Here, we use it from the perspective of polynomial degree-sensitivity to obtain results in the domain of 
functions with non-constant order of sensitivity. The layout of this section is as follows. 
We first discuss the Maiorana-McFarland construction and then obtain separation between $n$ and 
$\pdeg(f)$ while making the function first order sensitive (fully sensitive). We extend this 
construction while discussing first order sensitivity only for ease of understanding. 
We analyze the polynomial structure of these functions and obtain logarithmic separation between $n$ and $\pdeg(f)$. 
Finally we show that this construction can be modified for super-constant orders of sensitivity (upto $o(\log n)$ )
with only a few tweaks. 

\subsection{Maiorana-McFarland construction}
The Maiorana-McFarland (MM) construction~\cite{MM1} is based on dividing the input variable space into two parts
and attaching different linear functions from one subspace to each point in the other subspace, defined as follows. 

\begin{definition}
A function $f:\F^n \rightarrow \F$ is called an MM function 
if it can be expressed as
$f(\xv,\yv)=\left( \phi(\xv) \cdot \yv \right) \op g(\xv)$, where
\begin{itemize}

\item $\xv \in \F^{n_1}$ , $\yv \in \F^{n_2}$ , $n=n_1+n_2;$

\item $\phi$ is a mapping of the form $\phi: \F^{n_1} \rightarrow \F^{n_2};$

\item $g: \F^{n_1}\rightarrow \F$ is an arbitrary Boolean function defined on the subspace $\F^{n_1}$.
\end{itemize}
\end{definition}
The Boolean functions due to Maiorana-McFarland construction can be visualized in different ways. 
We view them as different linear functions defined on $\yv$ attached to activating values in $\xv$. 
Let there be an MM Boolean function with any arbitrary map $\phi: \F^{n_1} \rightarrow \F^{n_2} $ 
and some Boolean function $g: \F^{n_1} \rightarrow \F$. 
Corresponding to any $\var{a} \in F^{n_1}$, the quantity $\phi(\var{a}) \cdot \yv$ 
is essentially the outcome of the linear equation $\displaystyle \bigoplus_{\phi(\var{a})_i=1} \yv_i$. Thus $\phi(\var{a}) \cdot \yv \op g(\var{a})$ equals  
$\left( \displaystyle \bigoplus_{\phi(\var{a})_i=1} y_i  \right) \op g(\var{a})$ for all $\yv$ in $\F^{n_2}$.
Let us denote this function defined on $\F^{n_2}$ as $Lin_{\phi(\var{a}),g(\var{a})}$.
We now describe two real polynomial structures.
\begin{enumerate}

\item $\pr_{\var{a}}:\F^{n_1} \rightarrow \mathbb{R}, ~ \var{a} \in \F^{n_1}$ 
is defined as $\pr_{\var{a}}(\xv)= \left( \prod\limits_{\var{a}_i=0}(1 - x_i) \right)
\left( \prod\limits_{\var{a}_i=1} x_i \right) $, so that
$\pr_{\var{a}}(\xv) = 1 \text{ if $\xv=\var{a}$ and $0$ otherwise }$

\item[]

\item $Lin_{\phi({\var{a}}),g(\var{a})}$
corresponding to each $\var{a} \in \F^{n_1}$.
Any linear function on $\yv$ can be expressed as  
$\oplus_{b_i=1} y_i \op c, \var{b} \in \F^{n_2}, c \in \F$.
Then the 
$\lin_{(\var{b},c)}(\yv) = 
\frac{1}{2} -  \frac{ (-1)^c}{2}   \prod\limits_{b_i=1} (1-2y_i)$.

\end{enumerate}
Then we have the following real polynomial w.r.t to any MM type function.

\begin{proposition}
\label{prop:1} 
Given an MM function $f(\xv,\yv)= ( \phi(\xv) \cdot \yv ) \op g(\xv)$ on $n$ variables 
with $\xv \in \F^{n_1}$ and $\yv \in \F^{n_2}$,
the corresponding real polynomial can be defined as
\begin{equation}
\label{eq:2}
p(\xv,\yv)=  \displaystyle \sum_{\var{a} \in \F^{n_1}} 
\pr_{\var{a}}(\xv) \,
 \lin_{(\phi(\var{a}),g(\var{a}))} (\yv).
\end{equation}
\end{proposition}

Let us now note down a simple result that this polynomial structure entails. 

\begin{note}
\label{note:2}
For any $\var{b} \in \F^{n_2}$ we have  $\lin_{\var{b},c} + \lin_{\var{b},\overline{c}} =1 $ and
$\pdeg(\lin_{\var{b},c})=wt(\var{b})$.
\end{note}

The structure of the rest of this section is  as follows. 
First we describe some sufficient condition that allows a MM type function to have $s(f)=n$.
Next we obtain a MM type function with $n-1$ polynomial degree and then extend this technique to obtain a function with $\pdeg(f)=n-\log n$. Finally we extend this notion to higher order
sensitivity by adding non-linear functions on $\F^{n_2}$, which is one of the main results of the paper.

We first ensure $s(f)=n$. We add two restriction to a MM type function.
\begin{enumerate}
\begin{multicols}{2}
\item 
$\phi(\bm{1}_{n_1})=\bm{1}_{n_2}$ 
and 
$g(\bm{1}_{n_1})=0$.

\columnbreak

\item 
$wt \left( \phi(\bm{1}_{n_1}^i) \right) \equiv 1 \bmod 2$ 
and 
$g(\bm{1}_{n_1}^i) \equiv n_2  \bmod 2$.

\end{multicols}
\end{enumerate}
Then for all such functions we have $s(f)=n$. We denote this class of MM functions as $\M_n$. 
Refer to the formal presentation in Lemma~\ref{lemma:1} in Appendix~\ref{app:sec:4:1}. Now we describe the construction for the first separation. 

\subsection{First Order Sensitive Functions With Lower Polynomial Degree \label{sec:4:2}}

First we show a construction for getting $\pdeg(f)=n-1$. 
\begin{construction}
\label{th:1:c}
Let $f \in \M_n$ be an MM function defined on $n=n_1+n_2$ variables 
so that $n_1 \leq n_2 \leq n_1+1$  with $n_2 \equiv 0 \bmod 2$.
If $f$ is defined using $\phi(\bm{1}_{n_1-2}00)  = \bm{1}_{n_2}$, $g(\xv)= \prod_{i=1}^{n_1-2}x_i(1-x_{n-1})(1-{x_n})$, 
the sensitivity of $f$ is $n$ and the polynomial degree is at most $n-1$.
Refer to Appendix~\ref{app:sec:4:2} for the proof, along with examples and a count on the number of such functions. 
\end{construction}

Let us now better understand how the polynomial structure of the MM type functions can be modified so that the modified polynomial still represents a Boolean function, but with lower polynomial degree. 

\subsection{Interpreting real polynomial terms via the MM construction}

We saw in Proposition~\ref{prop:1} 
that the real polynomial corresponding to any MM type function
can be expressed as 
$
p(\xv,\yv)=    \sum_{\var{a} \in \F^{n_1}} 
\pr_{\var{a}}(\xv) 
\left(\lin_{(\phi(\var{a}),g(\var{a}))} (\yv) \right).
$
We know that $  \pr_{\var{a}}(\xv)= \left(\prod_{i: a_i=1} x_i\right) \left(\prod_{j: a_j=0}(1-x_j)\right)$.

If  $\var{z}=(z_1, \ldots ,z_k) \in \F^{k}$,
it is easy to see  
$
  \sum_{\var{a} \in \F^{k}}
\left( \big( \prod_{i: a_i=0} z_i \big) \big( \prod_{j: a_j=0}(1-z_j)\big) \right)=1.
$
Corresponding to a Boolean function defined on $n$ variables 
$\xv= (x_1,x_2,\ldots ,x_n)$,
we define three non-empty mutually disjoint sets 
$S_1,S_2$ and $S_3$ such that $S_1 \cup S_2 \cup S_3=[n]$
with $\abs{S_i}=s_i$.
Let the variables indexed by elements in $S_i$ be denoted as $x_{i_j}$
and $\var{z} \in \F^{s_3}$ be represented as $(z_1,z_2,\ldots ,z_{s_3})$.
Then the real polynomial 
$  \left( \prod_{i \in S_1} x_i \right) \left( \prod_{j \in S_2} (1-x_j) \right)$
can be represented as 
$\left( \prod_{i \in S_1} x_i \right) \left( \prod_{j \in S_2} (1-x_j) \right)
\left(   \sum_{\var{z} \in \F^{s_3}} \left( \big( \prod_{i: z_i=0} x_{3_i} \big) \big( \prod_{j: z_j=0}(1-x_{3_j})\big) \right) \right)
$
 that is, 
$
\sum_{\var{z} \in \F^{s_3}}
\left(
( \prod_{i \in S_1} x_i )( \prod_{j \in S_2} (1-x_j))
(  \prod_{i: z_i=0} x_{3_i} ) ( \prod_{j: z_j=0}(1-x_{3_j})) \right)
$.
This implies that corresponding to an MM type function defined on $n=n_1+n_2$
variables, the polynomial 
$ 
\lin_{(\var{b},c)}(\yv) 
( \prod_{i = 1}^{k_1} x_i ) ( \prod_{j=k_1+1}^{k_2} (1-x_j) ) 
$
can be interpreted as 
$$
  \sum_{\var{a} \in \F^{n_1-k_1-k_2}}
\left( 
( \prod_{i=1}^{k_1} x_i ) ( \prod_{j=k_1+1}^{k_2} (1-x_j) ) 
( \prod_{a_i=1} x_{k_1+k_2+i} ) ( \prod_{a_j=0} (1-x_{k_1+k_2+j}) ) 
\right) 
\lin_{(\var{b},c)}(\yv).
$$
We represent it as 
$
  \sum_{ \var{a} \in \F^{n_1-k_1-k_2}} 
\pr_{(\bm{1}_{k_1}\bm{0}_{k_2}\var{a})}(\xv)\lin_{(\var{b},c)}(\yv).
$
Thus, we get 
\begin{equation}
\label{eq:4}
\left( 
( \prod_{i=1}^{k_1} x_i ) ( \prod_{j=k_1+1}^{k_2} (1-x_j) ) 
\right) 
\lin_{(\var{b},c)}(\yv)
=
  \sum_{ \var{t} \in \F^{n_1-k_1-k_2}} \left(
\pr_{(\bm{1}_{k_1}\bm{0}_{k_2}\var{t})}(\xv)\lin_{(\var{b},c)}(\yv) \right).
\end{equation}

These considerations imply the following result.
\begin{proposition}
\label{prop:2}
Let $f$ be an MM function 
$f(\xv,\yv)=\xv \cdot \phi(\yv)+g(\xv)$
such that if $\xv=\bm{1}_{k_1}\bm{0}_{k_2}\var{t}$ then $\phi(\xv)=0$
and $g( \bm{1}_{k_1}\bm{0}_{k_2}\var{t})=0,~ \forall \var{t} \in \F^{n_1-k_1-k_2}$.
Then the polynomial corresponding to the Boolean function $f$ can be written as
$
p(\xv,\yv)=    \sum_{\var{a} \in \F^{n_1}, \var{a} \neq \bm{1}_{k_1}\bm{0}_{k_2}\var{t} } 
\pr_{\var{a}}(\xv) 
\left(\lin_{(\phi(\var{a}),g(\var{a}))} (\yv) \right),
$
and the polynomial 
$p'(\xv,\yv)=p(\xv,\yv)+ \left( ( \prod_{i=1}^{k_1} x_i ) ( \prod_{j=k_1+1}^{k_2} (1-x_j) ) 
\right) 
\lin_{(\var{b},c)}(\yv)$
can be written as
$
p'(\xv,\yv)=
  \sum_{\var{a} \in \F^{n_1}} 
\pr_{\var{a}}(\xv) \,
\lin_{(\phi'(\var{a}),g'(\var{a}))} (\yv) 
$,
where
\begin{multicols}{2}
\begin{align*}
  \phi'(\xv)&=
    \begin{cases}
      \phi(\xv) & \text{if $\xv \neq \bm{1}_{k_1}\bm{0}_{k_2} \var{t}$ for some $\var{t}$}\\
      b & \text{otherwise,}
    \end{cases}     
\end{align*}

\columnbreak

\begin{align*}
    g'(\xv)&=
    \begin{cases}
      g(\xv) & \text{if $\xv \neq \bm{1}_{k_1}\bm{0}_{k_2} \var{t}$ for some $\var{t}$}\\
      c & \text{otherwise}
    \end{cases}
\end{align*}
\end{multicols}
and this represents another MM type function $f'$ which differs from 
$f$ only in the points $\bm{1}_{k_1}\bm{0}_{k_2}\var{t}$.
\end{proposition}
Using this result we can attempt to obtain an MM type Boolean function with 
a pre-decided real polynomial structure. We start with the real polynomial 
of a particular MM type function, and then modifying its corresponding 
polynomial by adding   
$  \left( ( \prod_{i \in_1}^{k_1} x_i ) ( \prod_{j=k_1+1}^{k_2} (1-x_j) ) 
\right)  \lin_{(\var{b},c)}(\yv),$
keeping in mind the respective necessary constraints we have discussed 
in terms of $\phi$ and $g$.
This gives us another function $f'$ whose structure and 
its properties can be recovered from $f$.
Using this combinatorial approach we next have the following result.
\begin{result}
\label{thm:3:r}
There exists a  Boolean function $f \in \M_n$ 
with $\pdeg(f)=n -\Theta (\log n)$.
\end{result}
One can refer to Appendix~\ref{log-sep} for the buildup, along with the proof.
Finally we extend our constructions and results for super-constant orders of sensitivity. 

\subsection{Extending to super-constant higher order sensitivity via the MM construction}

We have so far observed the situation where we have defined a function
on $n$ variables in the MM class as 
$f(\xv,\yv)=\left( \phi(\xv) \cdot \yv \right) \op g(\xv)$ where $\xv \in \F^{n_1}$
and $\yv \in \F^{n_2}$ with $n_1+n_2=n$. The simplest interpretation is 
choosing a linear function in $\yv$ (or its complement depending on $g$) 
corresponding to each point $\xv \in \F^{n_1}$.
We can extend this to nonlinear functions in $\yv$ being fixed with respect to the points in $\yv$, with $g_{\var{a}}$ being the non-linear function on $\yv$ to be evaluated when $\xv=\var{a}$.
Then the real polynomial corresponding to the function $f(\xv,\yv)$ can be written as
$p(\xv,\yv)= 
\displaystyle\sum_{\var{a} \in \F^{n_1}}
\pr_{\var{a}}(\xv) \hat{g}_{\var{a}}(\yv)
$
where 
$\hat{g}_{\var{a}} : \F^{n_2} \rightarrow \mathbb{R}$ 
is the real polynomial corresponding to the function $g_{\var{a}}$.
Now let us discuss some sufficient conditions to obtain $k$-th order sensitivity 
by choosing the proper $g_{\var{a}}$ functions.

\subsection{Obtaining $k$-th order sensitivity an reducing polynomial degree:}

We start by defining a function in $\F^{n_2}$ that is $k$-th order sensitive itself. 
We define this function as ${\sf sym}_{n_2}^k: \F^{n_2} \rightarrow \F, k \leq n$. 
The algebraic normal form of the function contains all degree $i, 1 \leq i \leq k$
monomials.  
For an example 
${\sf sym}^2_4(\yv)= 
y_1 \op y_2 \op y_3 \op y_4 \op
y_1y_2 \op y_1y_3 \op y_1y_4 
\op y_2y_3 \op y_2y_4 \op y_3y_4$.
Next we observe the sensitivity order of this function.

\begin{lemma}
\label{lem:sym}
The function ${\sf sym}^k_{m}$ is a function defined on $m$ variables.
$k$-th order sensitive around the all zero input point $\bm{0}_m$. 
\end{lemma}
The proof can be found in Section~\ref{app:lem:sym}. 
Next we define an MM type function $\M^k_n$ with nonlinear functions in $\yv$, which is $k$-th order sensitive. 
\begin{construction}
\label{lem:korder:c}
Any function $f: \F^{n_1+n_2} \rightarrow \F$ with the  algebraic normal form
$\displaystyle f(\xv,\yv)=\bigoplus_{\var{a} \in \F^{n_1}} \left( ac_{\var{a}}(\xv) \cdot g_{\var{a}}(\yv) \right),$
where $g_{\bm{1}_{n_1}}= {\sf sym}^k_{n_2}$ 
and $g_{\var{a}}= 1$ for all $\var{a} \in \F^{n_1} : n_1 >wt(\var{a}) \geq n_1-k$
is $k$-th order sensitive.
\end{construction}
Finally we extend the technique of Section~\ref{log-sep} to obtain 
non-constant separation between number of variables and real polynomial 
degree in functions with super-constant order of sensitivity. 
\begin{construction}
\label{th:f10:c}
There exists a $k$-th order sensitive function in $\M_n^k$ with $n- \frac{\log \left( \frac{n}{2} -k \right)- \log{k}}{k}$
real polynomial degree.
\end{construction}
The formal proofs of Constructions~\ref{lem:korder:c} and~\ref{th:f10:c}
can be found in Appendix~\ref{mm:sup1}.
To reflect on the implications of this result, 
we can have $\log \log n$-order sensitive functions with 
$\pdeg(f) \approx n- \frac{\log n}{\log \log n}$. In fact as long as 
$k=o(\log n)$, we have $\pdeg(f)=n-\omega(1)$. These results could not 
be achieved via the recursive amplification constructions.

\section{Conclusion \label{sec:5}}
In this paper we have studied the interplay of resiliency and polynomial with respect to sensitivity, and have also extend the notion of sensitivity to higher order sensitivity. In this direction based on properties of resilient functions, we have obtained new classes of $6$-variable first order sensitive functions with $\pdeg(f)=3$, while also obtaining the same result for second order sensitivity. Which indicates that the function of minimum polynomial degree vs. sensitivity order and may not be a strictly increasing functions. 

Next we have studied the recursive amplification method and have designed slight modifications that allow us to start with base function, removing the restrictions of the simple function composition method. Furthermore, we use the resiliency-polynomial degree connection to design efficient circuits with linear size and logarithmic depth to realize highly-resilient ($n-o(n)$) functions. Our result improves on the best result known in this domain.

Finally we observe that for constant orders of sensitivity, we can have functions with $o(n)$ polynomial degree using the recursive amplification method. Against this backdrop we take the MM constructions and first obtain first order sensitive function with $\pdeg(f)=n-\Theta (\log n)$. Then we modify the MM construction with nonlinear function concatenation, and obtain functions with 
$n-\omega(1)$ polynomial degree and $\omega(1)$ order sensitivity.
Specifically, we show construction of $k$-th order sensitive function with $\pdeg(f)=n- \frac{\log \left( \frac{n}{2} -k \right)- \log{k}}{k}$. Our results enrich the domain of cryptographically important Boolean functions as long as lay down important combinatorial problems that should further enhance our understanding of real polynomial degree of Boolean functions.

\appendix

\section{Notes on Search Based Results \label{app:sec:2}}

\subsection{Nonexistence of $7$ variable fully sensitive function with Polynomial Degree Value of $3$ \label{app:713}}

We know that, for any $n$-variable Boolean function $f$,  $W_f(\var{u})=0$ for every $\var{u} \in \F^n$ with 
$wt(\var{u}) \geq n-m$ if and only if  $pdeg(f)\leq n-m-1$.
Thus, $f$ is an $n$-variable Boolean function with full sensitivity and polynomial degree $k$ if and only if 
$f \oplus 1$ is an $n$-variable Boolean function with full sensitivity and polynomial degree $k$.
Therefore, if there are $7$-variable Boolean functions with sensitivity $7$ and polynomial degree $3$,
there exist a $7$-variable Boolean function $f$ with sensitivity $7$, polynomial degree $3$, and 
a point $\xv \in \F^7$ such that $f(\xv)=0$ and $f(\xv \oplus e_i)=1$ for all $1 \leq i \leq 7$ where $e_i$'s are pairwise distinct vectors with Hamming weight $1$.
Assume that $f(\overline{i})=y_i$, where $\overline{i}$ is the binary expansion of $i$, and $y_{i'}=0$ when $\overline{i'}=\xv$ and $y_{i_1}=y_{i_2}=\ldots={y_{8}}=1$
when $e_j=\overline{i_j}$, then we can search if there exist such functions with sensitivity $7$ and polynomial degree $3$.
\begin{eqnarray*}
\begin{array}{llll}
\text{Minimize}& y_1&\\
\text{subject to} && \\
&\sum_{i=1}^{128}(1-2y_i)(-1)^{\var{u} \cdot \overline{i}}=0& wt(\var{u})\geq 4\\
&y_i\in\{0,1\}&1\leq i \leq 128 \\
&y_{i'}=0& \\
&y_{i_1}=y_{i_2}=\ldots={y_{8}}=1& 
 \end{array}.
\end{eqnarray*}

With the help of Gurobi, we checked for every $\xv \in \F^7$ and 
received a negative outcome in all cases, concluding that there is no $7$ variable fully sensitive function with $\pdeg(f)=3$.

\subsection{Rotation symmetric function for up to $10$ variables \label{app:rot}}
First we check the rotation symmetric functions on $7$ variables to obtain functions 
of third order sensitivity and then study rotation symmetric functions on $ns>7$ variables. 
We thus obtain the following results:

\begin{itemize}

\item There exists $12$ many $[7,1,3]$-functions (respectively $(7,5,3)$-functions). 
Recursively amplifying this function gives us a $(7^u, 5^u, 3)$-function, an instance
of $\pdeg(f)=n^{\frac{\log 7}{\log 5}}$. This is the best separation we are able to find for
third order sensitive functions. It will be interesting to observe if one can obtain better 
separation in this case.

\item There exists only $12$ functions in the rotation symmetric class that are $[8,1,3]$-functions
and none are $[8,2,3]$-functions. 
There cannot exist any $[8,\ell,4]$-function, where $\ell>0$.

\item We also find $29$ many $[9,1,4]$ functions rotation symmetric functions, out of which $27$ are also $[9,2,4]$-functions.  
Furthermore, there cannot exist any $[9,\ell,5]$-function, where $\ell>0$. Here one should note that one can also obtain 
$[9,1,4]$ and $[9,2,4]$-functions by recursively amplifying a $(3,1,2)$ or a $(3,2,2)$-function respectively 
and then taking its dual, and these were the only known $(9,1,4)$-functions before now. 
However, we obtain  $[9,1,4]$-rotation symmetric functions with a nonlinearity of $224$, 
where as the nonlinearity of the functions obtained through recursive amplification is $192$.
Thus, we obtain previously unknown $(9,1,4)$-functions. The advantage of using the recursive amplification process is its efficient circuit size and depth. 

\item There does not exist any $[10,1,5]$-rotation symmetric function. It should be noted that 
if we can obtain a $[10,1,5]$-function (provided such a function exists) then that would 
improve on the best known separation between sensitivity and polynomial degree. This is because
we can then get a $(10,1,4)$-function and then recursively amplify the function using the 
modified amplification process described in Theorem~\ref{th:ranew} to get a $(10^u,1,4^u)$-function, thus giving $s(f)=\left(\pdeg(f)\right)^{\frac{\log 10}{\log 4}}$ and this would be an improvement on the best known result. Furthermore we took all $(9,1,4)$-functions that we constructed and
used the reverse construction~\cite{SM03} where we concatenate the reverse of the truth table of an even
resilient function to itself to get a function with one more variable and one more resiliency, but this construction only gave us $(10,0,5)$-function and the sensitivity was not maintained.
\end{itemize}

This concludes the study of sensitivity-polynomial degree (and higher order sensitivity-polynomial degree) study
of functions on up to $10$ variables. 

\section{Results in Section~\ref{sec:3} \label{app:sec:3} }

\begin{theorem}
\label{th:ranew}
Let $f$ be a $k$-th order sensitive function on $d$ variables with $\pdeg(f_1)=p$ and 
$\yv=(y_1,y_2, \ldots, y_d)$ being the input with respect to which the function exhibits $k$-th order sensitivity.
We define the function  $f^u$ on $d^u$ variables such that:
\begin{enumerate}

\item ${\yv}^1=\yv$

\item ${\yv}^{i} \in \F^{d^i}$ is obtained by concatenating $d$ copies of  ${\yv}^{i-1}$.

\item $f^1=f$

\item 
$
f^i= f\Big( f^{i-1}(x_1,\ldots ,x_{d^{i-1}}) \op f^{i-1}(\yv^{i-1}) \op y_1
,\ldots, 
f^{i-1}(x_{jd^{i-1}+1}, \ldots ,x_{(j+1)d^{i-1}} ) \op f^{i-1}(\yv^{i-1}) \op y_j
, \ldots, 
f^{i-1}(x_{(d-1)d^{i-1}+1}, \ldots, x_{d^i}) \op f^{i-1}(\yv^{i-1}) \op y_d \Big).
$
\end{enumerate}
Then $f^u$ is a $k$-th order sensitive function on $n=d^u$ variables and $\pdeg(f)=p^u$
with $k$-th order sensitivity achieved at the input point ${\yv}^u$.
\end{theorem}
\begin{proof}
Here we call the function $f$ as the base function.
Let us denote by $[\xv]_k$ any input point 
that can be obtained by flipping at least $1$ and at most $k$ variables 
of $\xv \in \F^n$.
Thus if a function $f$ is $k$-th order sensitive at the point $\yv$ 
then $f([\yv]_k)=\overline{f(\yv)}$ by definition.
We now prove the result using induction on $u$.
The result holds for $u=1$ by definition.
Assume the result holds for $u-1$ and
we need to show that the function $f^u$ has $k$-th order sensitivity at $\yv^u$. 
The value of the function at $\yv^u$ is 
\begin{align*}
f^u(\yv^u)=&f\Big( f^{u-1}(\yv^{u-1}) \op f^{u-1}(\yv^{u-1}) \op y_1
,\ldots
,f^{u-1}(\yv^{u-1}) \op f^{u-1}(\yv^{u-1}) \op y_d \Big)
\\=&
f(\yv).
\end{align*}

Let us now select any $i \leq k$ variables whose value we wish to flip
resulting in an input point of the form ${[\yv^u]}_k$. 
We define the $d$ tuple $S=(s_1,s_2,\ldots, s_d)$ where $s_i$ denotes the 
number of bits to be flipped between $x_{(i-1)d^{u-1}+1}$ and $x_{id^{u-1}}$.
Thus $0 \leq s_i \leq k,~\forall i$.
If $s_i=0$ then 
\begin{align*}
&f^{u-1}(x_{id^{u-1}+1}, \ldots ,x_{(i+1)d^{u-1}} ) \op f^{u-1}(\yv^{u-1}) \op a_i
\\&=f^{u-1}(\yv^{u-1}) \op f^{n-1}(\yv^{u-1}) \op a_i=
a_i.
\end{align*}
If $1 \leq s_i \leq k$
then 
\begin{align*}
&f^{u-1}(x_{id^{u-1}+1}, \ldots ,x_{(i+1)d^{u-1}} ) \op f^{u-1}(\yv^{u-1}) \op a_i
\\&=f^{u-1}([\yv^{u-1}]_k) \op f^{u-1}(y^{u-1}) \op a_i=
\overline{a_i}.
\end{align*}
The number of nonzero values in $S$ are at most $k$, which would change 
at most $k$ of the $d$ points $y_i$ in the base function's input to $\overline{y_i}$
and result in an input to $f$ of the form of $f([y]_k)$. 
Thus for the function $f^u$ we have 
$
f^u([\yv^u]_k)=f([\yv]_k)=\overline{f(\yv)}=\overline{f^u(\yv^u)}
$.

The polynomial degree result holds from the basic definition of recursive amplification
as $\pdeg(f)=\pdeg(\overline{f})$ and this completes the proof.
\end{proof}

\begin{theorem}
\label{th:recres}
Given a function $f$ on $d$ variables with $\pdeg(f)=p<d$ we can obtain a function 
on $g^u$ on $n=d^u$ variables with resiliency $n-n^{\frac{\log p}{\log d}}-1$ such that 
there is a circuit of linear size and logarithmic depth (in $n$) for it.
\ \\
Here $g^u$ is the dual of $f^u$ where $f^u$ is the function on $n^u$ variables obtained by recursively amplifying $f$.
\end{theorem}

\begin{proof}
We first define $f^u$ as the function obtained recursively amplifying the function $f$, $u$ times,
which gives us a function on $d^u$ variables with $\pdeg(f^u)=n^{\frac{\log p}{\log d}}$.
Let us assume the circuit corresponding to the base function on $d$ variables consists of some $c_d$ gates
and has a depth of $t_d$. This circuit takes in $d$ input variable bits and outputs a single bit. 
Then the circuit corresponding to $f^u$ can be built using the circuits for $f$ in a layered manner 
in the following way. 

\begin{itemize}

\item In the first layer there are total $d^{u-1}$ circuits each taking in $d$ variables each as input bits. 

\item In the $i$-th layer there $d^{u-i-1}$ circuits each taking as input $d$ of the
$d^{u-i}$ output bits from the previous layer.

\item The final layer contains a single circuit, whose output is the output of the final function. 

\end{itemize}

Then the total number of circuit instances of $f$ to be used is 
$
\sum_{i=0}^{u-1} d^i= \frac{d^u-1}{u-1}
$
and the gate count is $c_d \times  \frac{d^u-1}{u-1}= \mathcal{O}(d^u)=\mathcal{O}(n)$.
Moreover, the depth of this circuit is $u \times t_d$ as the circuit for $f$ is set up in $u$ layers, 
which gives as a circuit for $f^u$ with $\mathcal{O}(\log_d n)$ depth.

Now if we XOR the parity of all the input bits to this output we obtain a 
$n-n^{\frac{\log p}{\log d}}-1$ function $g^u$ via the resiliency-polynomial degree connection.
The parity of the input bits can be simply obtained in parallel using $n$ gates and $\log n$ depth, which gives us the result.
\end{proof}

\begin{figure}[H]
\begin{center}
\includegraphics[scale=0.225]{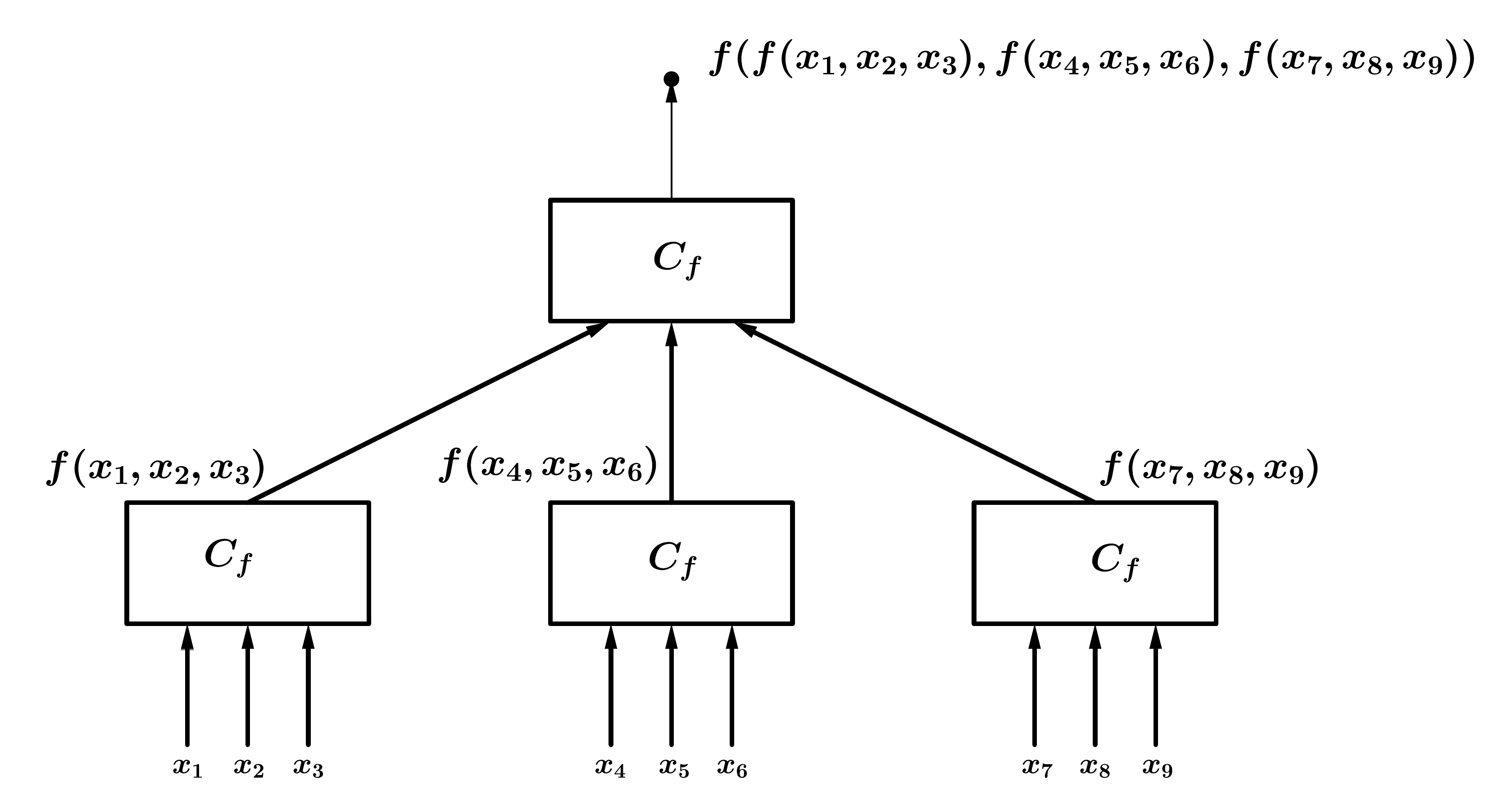}
\caption{Example of a circuit corresponding to recursive amplification.}
\label{fig:rec}
\end{center}
\end{figure}

\section{Cryptographic Properties of The Highly Resilient Functions \label{app:crypto-prop}}

\subsection*{Other Properties:}
Having discussed the efficiency of this method, let us now look into the nonlinearity of such functions,
which is another very important cryptographic property. We consider the following examples:
\begin{itemize}

\item Recursively amplifying a $(3,1,2)$-variable function $f_3$ to get a $9$ variable function 
$f_9(\xv)=f_3\big(f_3(x_1,x_2,x_3) \op a_1, f_3(x_4,x_5,x_6) \op a_2, f_3(x_7,x_8,x_9) \op a_3 \big),~ a_i \in \{0,1\}$
and then adding the all variable linear function to obtain $g_9$, which is $4$ resilient, irrespective of the choice of $a_i$.
However, depending on the choice of $a_i$, the nonlinearity can either be $96$ or $192$.
Here the circuit for $f_3$ needs $5$ XOR gates and $3$ AND gates and thus the circuit for $f_9$ 
only requires $20$ XOR gates and $12$ AND gates, and we can obtain $g_9$ by adding the all variable 
parity function, which requires a further $9$ XOR gates, which makes the total gate count to be only 
$41$.

\item Recursively amplifying a $4$ variable function $f_4$ to get a $16$ variable function $f_{16}$
and then adding the all variable linear function to obtain $g_{16}$. Here we obtain $16$ variable 
$11$-resilient functions with nonlinearity as high as $24576= 2^{15}-2^{13}$ where the best possible
nonlinearity for $16$ variable functions is $2^{15}-2^7$ (bent functions).  

\item We also find $(6,1,3)$-functions in Section~\ref{sec:4-6} 
that require $24$ AND gates and $21$ XOR gates to compute. We can then get 
a $36$ variable $26$ resilient fully dual sensitive function with only
$168$ AND gates and $147$ XOR gates.
\end{itemize}

Let us now look into some results on non trivial lower bounds on the nonlinearity of functions based on the recursive amplification method.
Here it should be noted that the algebraic degree of the function $g^u$ 
is upper bounded by the real polynomial degree of $f^u$, which is $n^{\frac{\pdeg(f)}{d}}$.

Let $f$ be a Boolean function on $s$ variables and
$g_1,g_2,\ldots,g_s$ be $s$ Boolean functions on $u$ variables.
We now define a $us$-variable Boolean function $\hat{f}$ by
\begin{equation}\label{Eq:CascadedBF}
\hat{f}\left(x_1,\ldots, x_{us}\right)=f\left(g_1(x_1,\ldots,x_u), g_2(x_{u+1},\ldots,x_{2u}), \ldots, g_s(x_{us-u+1},\ldots,x_{us})\right).
\end{equation}

\begin{theorem}\label{T:CascadedBF}
Let $\hat{f}$ be the function defined in \eqref{Eq:CascadedBF}, then for any ${\bf w}=({\bf w}_1,\ldots,{\bf w}_s)\in\mathbb{F}_2^{us}$
(where ${\bf w}_i\in\mathbb{F}_2^u$ for any $1\leq i \leq s$), we have 
\begin{eqnarray*}
W_{\hat{f}}({\bf w})&=&2^{-s}\sum_{{\bf v}\in\mathbb{F}_2^s}\left(W_f({\bf v})\prod_{i=1}^sW_{v_ig_i}({\bf w}_i)\right),
\end{eqnarray*}
where ${\bf v}=(v_1,v_2,\ldots,v_s)\in\mathbb{F}_2^s$.
\end{theorem}
\begin{proof}
We denote by ${\bf x}$ the vector $(x_1, \ldots, x_{us})$ and
${\bf x}_i$ the vector $(x_{ui-u+1},  \ldots,x_{ui})$ for any $1\leq i \leq s$. Then 
we have ${\bf x}=({\bf x}_1,{\bf x}_2,\ldots,{\bf x}_s)$ and $\hat{f}$ 
can be rewritten as  $\hat{f}\left({\bf x}\right)=f\left(g_1({\bf x}_1), g_2({\bf x}_2), \ldots, g_s({\bf x}_s)\right)$.
For any ${\bf x}\in\mathbb{F}_2^{us}$, we define ${\bf g}=(g_1({\bf x}_1),g_2({\bf x}_2),\ldots,g_s({\bf x}_s))\in\mathbb{F}_2^s$.
Then for any ${\bf g}\in\mathbb{F}_2^s$, we have
\begin{eqnarray}\label{eq:linearfun}
\sum_{{\bf y}\in\mathbb{F}_2^s}\sum_{{\bf v}\in\mathbb{F}_2^s}(-1)^{{\bf v}\cdot({\bf y}+{\bf g})}=2^s,
\end{eqnarray}
since $\sum_{{\bf v}\in\mathbb{F}_2^s}(-1)^{{\bf z}\cdot {\bf v}}$ equals 0 if ${\bf z}\neq {\bf 0}$ and equals $2^s$ if ${\bf z}= {\bf 0}$.
By the definition of the Walsh transform, for any ${\bf w}=({\bf w}_1,{\bf w}_2,\ldots,{\bf w}_s)\in\mathbb{F}_2^{us}$ we have  
\begin{eqnarray*}
W_{\hat{f}}({\bf w})&=&\sum_{{\bf x}\in\mathbb{F}_2^{us}}(-1)^{\hat{f}({\bf x})\oplus{\bf w}\cdot{\bf x}}\\
&=&\sum_{{\bf x}\in\mathbb{F}_2^{us}}(-1)^{f({\bf g})\oplus{\bf w}\cdot{\bf x}}\\
&=&2^{-s}\sum_{{\bf x}\in\mathbb{F}_2^{us}}(-1)^{f({\bf y})\oplus{\bf w}\cdot{\bf x}}\sum_{{\bf y}\in\mathbb{F}_2^s}\sum_{{\bf v}\in\mathbb{F}_2^s}(-1)^{{\bf v}\cdot({\bf y}\oplus{\bf g})}~~(\text{by}~\eqref{eq:linearfun})\\
&=&2^{-s}\sum_{{\bf v}\in\mathbb{F}_2^s}\sum_{{\bf y}\in\mathbb{F}_2^s}\sum_{{\bf x}\in\mathbb{F}_2^{us}}(-1)^{f({\bf y})\oplus{\bf v}\cdot{\bf y}\oplus{\bf v}\cdot{\bf g}\oplus{\bf w}\cdot{\bf x}}\\
&=&2^{-s}\sum_{{\bf v}\in\mathbb{F}_2^s}\left[\left(\sum_{{\bf y}\in\mathbb{F}_2^s}(-1)^{f({\bf y})\oplus{\bf v}\cdot{\bf y}}\right)\left(\sum_{{\bf x}\in\mathbb{F}_2^{us}}(-1)^{{\bf v}\cdot{\bf g}\oplus{\bf w}\cdot{\bf x}}\right)\right]\\
&=&2^{-s}\sum_{{\bf v}\in\mathbb{F}_2^s}\left(W_f({\bf v})\sum_{{\bf x}_1,{\bf x}_2,\ldots,{\bf x}_s\in\mathbb{F}_2^u}(-1)^{\bigoplus_{i=1}^s v_ig({\bf x}_i)\oplus\bigoplus_{i=1}^s{\bf w}_i\cdot {\bf x}_i}\right)\\
&=&2^{-s}\sum_{{\bf v}\in\mathbb{F}_2^s}\left(W_f({\bf v})\prod_{i=1}^sW_{v_ig_i}({\bf w}_i)\right).
\end{eqnarray*}
\end{proof}

\begin{theorem}\label{T:CascadedBalacedBF}
Let $\hat{f}$ be the function defined in \eqref{Eq:CascadedBF} by taking $g_1,g_2,\ldots,g_k\in\mathbb{F}_2^k$ to be balanced functions.
Then for any ${\bf w}=({\bf w}_1,{\bf w}_2,\ldots,{\bf w}_s)\in\mathbb{F}_2^{ks}$ (where ${\bf w}_i\in\mathbb{F}_2^k$ for any $1\leq i \leq s$), we have
\begin{eqnarray*}
W_{\hat{f}}({\bf w})&=&2^{-s}W_f({\bf v}')\prod_{i=1}^s\left(\frac{1+(-1)^{v_i'}}{2}\cdot 2^k+\frac{1-(-1)^{v_i'}}{2}\cdot W_{g_i}({\bf w}_i)\right),
\end{eqnarray*}
where ${\bf v}'=(v_1',v_2',\ldots,v_s')\in\mathbb{F}_2^s$ with $v_i'=1$ if and only if ${\bf w}_i\neq {\bf 0}_k$.
Furthermore, we have 
$$NL(\hat{f})\geq 2^{ks-1}-2^{ks-k-s-1}\bigg(2^s-2NL(f)\bigg)\left(2^k-2\min_{1\leq i\leq s}\{NL(g_i)\}\right),$$
where $NL$ denotes the nonlinearity.
\end{theorem}
\begin{proof}
Note that the value of the Walsh transform at any nonzero point of a constant function is equal to null.
Thus, for any $1\leq i \leq s$, we have $W_{0}({\bf w}_i)=0$ if ${\bf w}_i\ne {\bf 0}_k$.
As $g_i$'s are balanced, we have $W_{g_i}({\bf 0}_k)=0$  for any $1\leq i \leq s$.
So we have $\prod_{i=1}^sW_{v_ig_i}({\bf w}_i)=0$ if $v_i\neq v_i'$.
Then by Theorem~\ref{T:CascadedBF} we have
\begin{eqnarray*}
W_{\hat{f}}({\bf w})&=&2^{-s}\sum_{{\bf v}\in\mathbb{F}_2^s}\left(W_f({\bf v})\prod_{i=1}^sW_{v_ig_i}({\bf w}_i)\right)\\
&=&2^{-s}W_f({\bf v'})\prod_{i=1}^sW_{v_i'g_i}({\bf w}_i).
\end{eqnarray*}
Then our first assertion comes from the fact $W_{v_i'g_i}({\bf w}_i)$ equals $2^k$ if $v_i'=0$ and $W_{g_i}({\bf w}_i)$ if $v_i'=1$. 
Note that $|W_{g_i}({\bf w}_i)|\leq \max_{{\bf w}_i'\in\mathbb{F}_2^k }\{ |W_{g_i}({\bf w}_i')| \}\leq 2^k$ for any $1\leq i \leq s$.
Then we have $|\prod_{i=1}^sW_{v_ig_i}({\bf w}_i)|\leq 2^{ks-k}\cdot \max_{1\leq i \leq s, {\bf w}_i'\in\mathbb{F}_2^k }\{ |W_{g_i}({\bf w}_i') |\}$  
by setting the Hamming weight of ${\bf v'}$ to be $1$. Then we can obtain that the nonlinearity of $\hat{f}$ is at least
$2^{ks-1}-2^{ks-k-s-1}\max_{{\bf u}\in\mathbb{F}_2^s}\{|W_{f}({\bf u})|\}\cdot \max_{1\leq i \leq s, {\bf w}_i'\in\mathbb{F}_2^k }\{ |W_{g_i}({\bf w}_i') |\}$,
which gives our second assertion.
This finishes the proof.
\end{proof}

\begin{lemma}\label{L:ff}
Let $f$ be an $n$-variables function and $\hat{f}$ be an $n^2$-variable function defined
as $\hat{f}(x_1,\ldots,x_{n^2})=f\left(f(x_1,\ldots,x_n)+a_1, \ldots, f(x_{n^2-n+1},\ldots,x_{n^2})+a_n\right)$, where $a_i$'s belong to $\mathbb{F}_2$.
Then the nonlinearity of $\hat{f}$ is at least $2^{n^2-n+1}NL(f)-2^{n^2-2n+1}\left(NL(f)\right)^2$.
Moreover, if $f$ is $m$ resilient, then $\hat{f}$ is $(m+1)^2-1$ resilient; if $w_f({\bf w})=0$ for any 
${\bf w}$ with Hamming weight no less than $m'$, then $f$ is $n^2-(m'-1)^2-1$ resilient.
\end{lemma}
\begin{proof}
The lower bound on nonlinearity of $\hat{f}$ directly follows from Theorem~\ref{T:CascadedBalacedBF}.
By observing the value of $W_f({\bf v'})\prod_{i=1}^sW_{v_i'g_i}({\bf w}_i)$ in the proof of Theorem~\ref{T:CascadedBalacedBF}, then 
we can easily obtain the rest assertions. This completes the proof.
\end{proof}

\begin{corollary}
Let $f$ be an $n$-variables Boolean function and $f_i$ ($i\geq 2$) be the function defined above.
Then we have $NL(f^{i})\geq 2^{n^i-n^{i-1}}NL(f^{i-1})+2^{n^i-n}NL(f)-2^{n^i-n^{i-1}-n+1}NL(f^{i-1})NL(f)$ with $NL(f^1)=NL(f).$
\end{corollary}

However, we can also have an algebraic degree-resiliency trade off in this construction if we use two kinds of 
base functions on $d$ variables, one with algebraic degree (and thus polynomial degree) of $d$ and the other having 
lower polynomial degree, and we can add these functions in different stages of the recursion to obtain functions 
with various resiliency and algebraic degree values. Thus we have two problems that we propose for further attention as they are beyond the scope of our current discussion:
\begin{enumerate}

\item How much can we increase the nonlinearity in  $f^u$ when the functions used in recursion are of the form 
$f$ or $\overline{f}$, where $f$ is the base function defined on some $d$ variables?

\item How good cryptographic profiles can we obtain using this method and using multiple functions on $d$ variables 
throughout the recursion process?
\end{enumerate}
 
\begin{remark} 
Thus, in summary, we start with a $d$ variable good resilient function $g$, take its dual $f$ 
and then recursively amplify to obtain $f^u$ on $d^u$ variable. Then we obtain $g^u=f^u \op \lin_{d^u}$
which is the final highly recursive function.

We can also have an algebraic degree-resiliency tradeoff. Algebraic degree is the maximum size of the monomials in the ANF of a function (the maximum number of variables in a product term). Higher algebraic degree is needed to protect a system against algebraic degree attack, where as high resiliency protects the system against correlation attack, two of the most powerful attacks in cryptanalysis of ciphers. 
In this construction we can also control whether the resultant function should be fully
sensitive (using Theorem~\ref{th:ranew}) or should it have lower sensitivity value and this does not affect the possible resiliency and algebraic degree values. 
\end{remark}

\subsection{Superlinear Separation Between $n$ and $\pdeg(f)$ For Constant Order Sensitivity \label{app:const-amp}}

First we observe how high can sensitivity order be for a function with less than $n$ polynomial degree. It is easy to see that if a function is $n$-th order sensitive then its real polynomial degree is also $n$. This is because a function can be 
$n$-th order sensitive if for a point $\var{a} \in \F^n$, $f(\var{a})=0$ and $f(\xv)=1,~\forall\xv \in \F^n \setminus \{\var{a}\}$. 
Then such a function will have an odd number of ones in its truth table and thus will have $\pdeg(f)=n$.

The first separation between $n$ and polynomial degree can be found for $(n-1)$-th order sensitivity when $n$ is odd, and for $(n-2)$-th order, otherwise. 
\begin{lemma}
\label{lem:first}	 
The maximum value of $k$ such that there exists a $k$-th order sensitive function with polynomial degree less than $n$ is $n-1$, where $n$ is odd, and $n-2$, otherwise. 
\end{lemma}
\begin{proof}
We prove this result via the resiliency idea. 
The dual of an $(n,k,n-1)$-function is an $[n,k,0]$-function, which is a balanced $k$-th order dual sensitive function.

Let $g$ be an $(n-1)$-th order dual sensitive function with respect to a point $\yv\in \F^n$. Without loss of generality we   assume   that $f(\yv)=0$.
Then, corresponding to all the points that can be obtained by flipping an odd number of bits, the output is $0$, and $1$, otherwise, where the maximum number bits that can be flipped is $n-1$.
That is, $g(\overline{\yv})$ can be both zero or one, which does not affect the sensitivity order, and the output corresponding to all other $2^n-1$ points is fixed. 

{\bf Case $n$  odd:}
The minimum number of points $\xv$ for which 
$f(\xv)$ must be zero is 
\begin{equation*}
1+ \binom{n}{1}+ \binom{n}{3}+ \cdots +\binom{n}{n-2}=2^{n-1}.
\end{equation*}

We can then fix $f(\yv)=1$ and the other $2^{n-1}-1$ points have output $1$, by definition,
which gives us an $(n-1)$-th order dual sensitive $0$-resilient function whose dual $f= g \op \lin_n$ is then an $(n,n-1,n-1)$-function.

{\bf Case $n$ even:}
Here, the minimum number of points $\xv$ for which $f(\xv)=0$ is
\begin{equation*}
1+ \binom{n}{1}+ \binom{n}{3}+ \cdots +\binom{n}{n-1}
=2^{n-1}+1
\end{equation*}
and thus the corresponding function cannot be balanced. 

On the other hand if we try to form an $(n-2)$-th order dual sensitive function, then 
the restriction reduces by $\binom{n}{n-1}$ and becomes $2^{n-1}-n+1$.
Meanwhile, the restriction on the minimum number of points with $f(\xv)=1$ remains the same and then the remaining $n-1$ points can be fixed accordingly to make the function balanced, which gives us an $[n,n-2,0]$-function whose dual is an $(n,n-2,n-1)$-function.
\end{proof} 

Then our result of super-linear separation between $n$ and $\pdeg()$ follows from our 
modified recursive amplification construction. 

\begin{theorem}
Given any constant $k$ there exists a $k$-th order sensitive function $f$ on $n$ variables such that 
$\pdeg(f)= n^{\frac{\log k}{\log k+1}}$, if $k$ is even and 
$\pdeg(f)= n^{\frac{\log k+1}{\log k+2}}$, if $k$ is odd.
\end{theorem}
\begin{proof}

Given any $k$, we form a function $f$ in the following manner:
\begin{itemize}
\item If $k$ is even we form a $(k+2)$-variable balanced $k$-th order dual sensitive
function $g$ which is a $[k+2,k,0]$-function as per Lemma~\ref{lem:first}.
Then we take its dual $f= g \op \lin_n$, which is a $(k+2,k,k+1)$-function.

\item If $k$ is odd we form a $(k+1)$-variable balanced $k$-th order dual sensitive 
function $g$ which is a $[k+1,k,0]$-function.
Then we take its dual $f= g \op \lin_n$, which is a $(k+1,k,k)$-function,
via Lemma~\ref{lem:first}.
\end{itemize}

Next we use the  recursive amplification process described in Theorem~\ref{th:ranew}.
This gives a $\left({(k+2)}^u,k,{(k+1)}^u\right)$-function, when $k$ is even and a $\left({(k+1)}^u,k,{k}^u\right)$-function, 
otherwise, which gives us the desired super-linear advantage.
\end{proof}

\section{Results in Section~\ref{sec:4} \label{app:sec:4} }

\subsection{Sufficient Conditions for $s(f)=n$ \label{app:sec:4:1}}

\begin{lemma}
\label{lemma:1}
Let us denote by $\M_n$ the set of MM type functions with the following 
conditions. 
\begin{enumerate}

\item 
$\phi(\bm{1}_{n_1})=\bm{1}_{n_2}$ 
and 
$g(\bm{1}_{n_1})=0$.

\item 
$wt \left( \phi(\bm{1}_{n_1}^i) \right) \equiv 1 \bmod 2$ 
and 
$g(\bm{1}_{n_1}^i) \equiv n_2  \bmod 2$.
\end{enumerate}
Then for any function $f \in \M_n$ we have $s(f)=n$.
\end{lemma}

\begin{proof}
It suffices to show that $s(f,\bm{1}_n)=n$. 
By definition we have $\phi(\bm{1}_{n_1})=\bm{1}_{n_2}$.
Thus for any input of the form $\bm{1}_{n_1}||\yv$ 
we have $f(\bm{1}_{n_1}||\yv)=   \bigoplus_{i=1}^{n_2} y_i$.
Then if 
$y_i=1$, for all $i$, we get $f(\bm{1}_{n_1}||\bm{1}_{n_2})= n_2 \bmod 2$. 
Corresponding to this point, if any one of the components in $\yv$ is flipped, 
then the function evaluates to $n_2+1 \bmod 2$.

If any of the points (say $x_i: 1 \leq i \leq n_1$) in $\xv$ is flipped then the output of the function is 
$f(\bm{1}_{n_1}^i||\bm{1}_{n_2})$. Since the number of variables of the linear functions is odd and $g(\bm{1}_{n_1}^i) \equiv n_2  \bmod 2$ 
we have $f(\bm{1}_{n_1}^i||\bm{1}_{n_2})=n_2 +1 \bmod 2$.
\end{proof}

\subsection{An MM function with $s(f)=n$ and $\pdeg(f) \leq n-1$ \label{app:sec:4:2} }

Let us first denote some notations that we shall use in the proof. 
We denote by  $\bf{X_i^j}$ 
the monomial $\left(   \prod_{k=i}^j x_k \right)$
and by $\bf{X^i}$ the monomial   $\left(   \prod_{k=1}^i x_k \right)$.
 We define the polynomial ${\bf L_{n_2}} =   \prod_{i=1}^{n_2} (1-2y_i).$
Then we can write $\lin_{{\bf 1_{n_2}},0}=\frac{1-\lf}{2}$ 
and $\lin_{{\bf 1_{n_2}},1}=\frac{1+\lf}{2}$.

\begin{theorem}
\label{th:1}
Let $f \in \M_n$ be an MM function defined on $n=n_1+n_2$ variables 
so that $n_1 \leq n_2 \leq n_1+1$  with $n_2 \equiv 0 \bmod 2$.
If $f$ is defined using $\phi(\bm{1}_{n_1-2}00)  = \bm{1}_{n_2}$, $g(\xv)= \prod_{i=1}^{n_1-2}x_i(1-x_{n-1})(1-{x_n})$, 
the sensitivity of $f$ is $n$ and the polynomial degree is at most $n-1$.
\end{theorem}

\begin{proof}
First we have $s(f)=n$, which follows directly  from the constraints described and the subsequent proof of Lemma~\ref{lemma:1}.
As $n_2$ is even, $g(\bm{1}_{n_1}||\bm{1}_{n_1})$ should be $0$, 
which it is, as the only point in $\xv$ for which $g(\xv||\yv)=1$
is $\bm{1}_{n_1-2}00$.
\\

We now show that $\pdeg(f)=n-1$.
Apart from the constraints given, we have two more constraints as $f \in \M_n$, namely, $wt(\phi(\bm{1}_{n_1}^i)) \equiv 1 \mod 2$, $\phi(\bm{1}_{n_1}) = \bm{1}_{n_2}$.
We have defined the real polynomial $p(\xv,\yv)$ that represents $f$ in Proposition~\ref{prop:1}.
The degree of the polynomial is   $\leq \max(wt(\phi))+n_1$.
We know that  
$p(\xv,\yv)=    \sum_{a \in \{0,1\}^{n_1}} \pr_{a}(\xv) \left(\lin_{(\phi(\var{a}),g(\var{a}))} (\yv) \right)$
has monomials of degree $n$ only in the terms $\lin_{b,c}$ where $wt(b)=n_2$, which is the case
only when $b=\bm{1}_{n_1-2}00$ or $\bm{1}_{n_1}$.
Additionally $g(\bm{1}_{n_1-2}00)=1$ and $g(\bm{1}_{n_1})=0$.
Therefore, we can write the polynomial as 
\begin{align*}
& p(\xv,\yv)=p'(\xv,\yv) + 
\pr_{(\bm{1}_{n_1-2}00)}(\xv) \lin_{\phi(\bm{1}_{n_1-2}00),g(\bm{1}_{n_1-2}00)}(\yv)
+
\pr_{(\bm{1}_{n_1})}(\xv) \lin_{\phi(\bm{1}_{n_1}),g(\bm{1}_{n_1})}(\yv) \\
& \text{ where }
p'(\xv,\yv)=
 \sum_{\var{a} \in \F^{n_1} \setminus~ \{ \bm{1}_{n_1},\bm{1}_{n_1-2}00 \} } 
\pr_{\var{a}}(\xv) 
\left(\lin_{(\phi(\var{a}),g(\var{a}))} (\yv) \right).
\end{align*}
Thus in this case $\pdeg(p')=n-1$. Now, expanding $p(\xv,\yv)$ we get
\begin{align*}
& p(\xv,\yv)=
\\&
p'(\xv,\yv) + 
\left(\prod\limits_{i=1}^{n_1-2} \xv_i \right)(1-x_{n_1-1})(1-x_{n_1}) 
\lin_{(\bm{1}_{n_2},1)}(\yv)
+
\left(\prod_{i=1}^{n_1}x_i \right)
\lin_{(\bm{1}_{n_2},0)}(\yv)
\\ & =
p'(\xv,\yv)+ {\bf X^{n_1-2}} (1-x_{n_1-1}-x_{n_1} + x_{n_1-1}x_{n_1}) 
\lin_{(\bm{1}_{n_2},1)}(\yv)
+
{\bf X^{n_1}}\lin_{(\bm{1}_{n_2},0)}(\yv)
\\ & =
p'(\xv,\yv)+
\left({\bf X_1^{n_1-2}} - {\bf X_1^{n_1-2}}x_{n_1-1} - {\bf X_1^{n_1-2}}x_{n_1} \right) \lin_{(\bm{1}_{n_2},1)}(\yv)
\\ & 
+ 
{\bf X^{n_1}} 
\left(
\lin_{(\bm{1}_{n_2},0)}(\yv)
+
\lin_{(\bm{1}_{n_2},1)}(\yv)
\right)
\\ & =
p'(\xv,\yv)+
\left({\bf X_1^{n_1-2}} - {\bf X_1^{n_1-2}}x_{n_1-1} - {\bf X_1^{n_1-2}}x_{n_1} \right) \lin_{(\bm{1}_{n_2},1)}(\yv)
+ 
{\bf X^{n_1}}.
\end{align*}
In this case, the polynomial degree of the terms are:
\begin{itemize}

\item $p'(\xv,\yv)$ has a degree of at most $n-1$;

\item 
$\left({\bf X_1^{n_1-2}} - {\bf X_1^{n_1-2}}x_{n_1-1} - {\bf X_1^{n_1-2}}x_{n_1} \right) \lin_{(\bm{1}_{n_2},1)}(\yv)
$ has a degree of $n-1$.

\item the degree of $X^{n_1}$ is $n_1$.
\end{itemize}

Thus, the polynomial degree of $p(\xv,\yv)$ is at most $n-1$ and the proof is complete.
\end{proof}

In this case, any function in $\M_n$ with the added constraint 
that $\phi(\bm{1}_{n_1-1}00)=\bm{1}_{n_2}$ and 
$g(\xv)= \left( \prod_{i=1}^{n_1-2} x_i \right)(1-x_{n_1-1})(1-x_n)$
would have this property of $s(f)=n$ and $\pdeg(f)=n-1$.  
Except for $\bm{1}_{n_2}00$ and $\bm{1}_{n_2}$, all the points can have any linear function
of weight less than $n_2$, with the only restriction being that $n_1$ points in $\xv$ 
should have odd weighted linear functions attached to them.
This gives us a loose lower bound on the number of functions
that we can recover with this construction technique with the given property.
We can in fact construct $2^{{\sf EXP}(n)}$ functions with such property, 
which we state next. Let us now look at an example function. 

We let $n=7$ and $n_1=3,n_2=4$ and show the structure of the Boolean function. 
Observe that for the inputs $001$, $011$ and $110$ in $\xv$ the corresponding linear function
contains only a single variable and the linear function corresponding to $111$ contains $4$ 
variables, implying that $f$ is in the class $\M_7$. 
Additionally the linear function at $100$ is the complement of the linear function at $111$,
which satisfies the conditions of Theorem~\ref{th:1}.
\begin{table}[H]
\centering
\begin{tabular}{|c|c|}
\hline 
$\xv$ & The linear function on  $\yv$ \\ \hline
$000$ &   $y_2 \op y_3$   \\ \hline
$001$ &   $y_1$    \\ \hline
$010$ &   $y_2$ \\ \hline
$011$ &   $y_1 \op y_2$ \\ \hline
$100$ &   $y_1 \op y_2 \op y_3 \op y_4 \op 1$ \\ \hline
$101$ &   $y_1 \op y_3$ \\ \hline
$110$ &   $y_3$ \\ \hline
$111$ &  $y_1 \op y_2 \op y_3 \op y_4$ \\ \hline

\end{tabular}

\caption{A Boolean function $f$ with $D(f)=7$ and $\pdeg(f)=6$.}
\label{tab:2}
\end{table}

\begin{corollary}
\label{cor:3}
For any $n$, there are at least 
$\Omega \left( 2^{2^{\frac{n}{2}}} \right)$ MM type functions with $D(f)=n$ and $\pdeg(f)=n-1$.
\end{corollary}
\begin{proof}
By definition we have $\lceil \frac{n}{2} \rceil \leq n_2 \lceil \frac{n}{2} \rceil +1$.
For simplicity, let us assume that $n$ is even. Then each distinct mapping $\phi$ that 
satisfies the constraint of Theorem~\ref{th:1} will represent an MM type function with $s(f)=n$ and $\pdeg(f)=n-1$.  
We have already defined $\phi$ at the points $\bm{1}_{n_1-2}00$ and $\bm{1}_{n_1}$. 
Let us consider the mapping where $\phi(\bm{1}_{n_1}^i)={0}_{n_2}^i$. 
That is, the function attached corresponding to the $i$-th critical point ($1 \leq i \leq n$)
is $y_i$ and since $n_2 >n_1$, this is feasible.
Therefore, corresponding to the other $2^{\frac{n}{2}}-2-\frac{n}{2}$ points in $\xv$ we can put any of the 
$2^{\frac{n}{2}}$ linear functions defined on $\yv$, and each such function would have the desired  property. 
 
Thus, the number of such functions will be  $\left( 2^{\frac{n}{2}}-2-\frac{n}{2} \right)^{2^{\frac{n}{2}}}$, which we denote by $C(n,n-1)$.
Then for $n>6$ we have 
$$
C(n,n-1)> \left(2^{\frac{n}{2}-1} \right)^{2^{\frac{n}{2}}} = \Omega \left( 2^{2^{\frac{n}{2}}} \right).
$$

\end{proof}

\subsection{Logarithmic separation Between $n$ and $\pdeg(f)$ with $s(f)=n$ \label{log-sep}}

We finally design a method of constructing an MM type function in $\M_n$ 
that has polynomial degree of $n- \Theta(\log n)$.
Before proceeding to the proof, we derive some 
more results related to the structure of different polynomials and their
interaction upon addition. 
In the last subsection we observed how given an MM type 
function $f$, where $\phi$ and $g$ are defined to be zero in many points,
say all points of the form $\bm{1}_{k_1}\bm{0}_{k_2}\var{t} \in \F^{n_1}$
then adding a polynomial term of the form 
$\left( \prod_{i=1}^{k_1}x_i \prod_{j=k_1+1}^{k_2} (1-x_j) \right) \lin_{\var{b},c} (\yv)$
transforms it to another MM type function $f'$ whose map $\phi'$ 
and sub-function on $g'$ differs from $f$ only for the points 
$\bm{1}_{k_1}\bm{00}_{k_2}\var{t} \in \F^{n_1}$.
With these observations in mind we analyze a function $f_1 \in \M_n$ 
such that its polynomial satisfies the constraint of Theorem~\ref{th:1} and 
then further add some constraints so as to be able to 
make modifications as discussed in Proposition~\ref{prop:2}.
We define the corresponding map $\phi_1$ and sub-function 
$g_1$ in the lines of  Theorem~\ref{th:1}:
\begin{equation}
\label{eq:f1}
\begin{split}
 \phi_1(\xv) &=
    \begin{cases}
      \bm{1}_{n_2} & \text{ if $\xv=\bm{1}_{n_1}$ or $\xv= \bm{1}_{n_1-2}00 $}\\
      \bm{0}_{n_2}^i & \text{ if $\xv=\bm{0}_{n_1}^i  $} \\
		0 & \text{ otherwise,}
    \end{cases}\\
  g_1(\xv)&= \left( \prod_{i=1}^{n_1-2}x_i \right)(1-x_{n_1-1})(1-x_{n_1}).         
  \end{split}
\end{equation}
This definition differs from the function in Theorem~\ref{th:1} in 
that the map $\phi_1$ is defined as $0$ in all but $n_1+2$ positions. 
However it satisfies all the constraints of Theorem~\ref{th:1}.
This means that the real polynomial corresponding to $f_0$ can be expressed as 
\begin{align*}
p_1(\xv,\yv)&
=p_1'(\xv,\yv)+
\left({\bf X_1^{n_1-2}} - {\bf X_1^{n_1-2}}x_{n_1-1} - {\bf X_1^{n_1-2}}x_{n_1} \right) \lin_{(\bm{1}_{n_2},1)}(\yv)
+ 
{\bf X^{n_1}} 
\\ &=
p_1'(\xv,\yv)+
\left({\bf X_1^{n_1-2}} - {\bf X_1^{n_1-2}}x_{n_1-1} - {\bf X_1^{n_1-2}}x_{n_1} \right)  \frac{1+ \lf}{2}
+ 
{\bf X^{n_1}},
\end{align*}
where $p_1'(\xv,\yv)$ defined in the same manner as $p'$ in Theorem~\ref{th:1}. Consider the term
$- {\bf X_1^{n_1-2}}x_{n_1-1} \frac{1+ \lf}{2}$ which has a degree of $n-1$.
If we add the polynomial 
${\bf X_1^{n_1-4}} (1-x_{n_1-3})(1-x_{n_1-2})x_{n_1-1}\frac{1-\lf}{2}$
to this term we get 
\begin{align*}
&{\bf X_1^{n_1-4}}(1-x_{n_1-3}-x_{n_1-2}+x_{n_1-3}x_{n_1-2})x_{n_1-1}\frac{1-\lf}{2}
-
{\bf X_1^{n_1-2}}x_{n_1-1} \frac{1+ \lf}{2}
\\&=
{\bf X_1^{n_1-4}}x_{n_1-1}
-{\bf X_1^{n_1-4}}x_{n_1-3}x_{n_1-1}\frac{1-\lf}{2}
-{\bf X_1^{n_1-4}}x_{n_1-2}x_{n_1-1}\frac{1-\lf}{2}
\\&+
{\bf X_1^{n_1-2}}x_{n_1-1}\frac{1-\lf}{2}
-{\bf X_1^{n_1-1}} \frac{1+ \lf}{2}
+{\bf X_1^{n_1-1}} \frac{1-\lf}{2}
\\
&=
\left( {\bf X_1^{n_1-4}}x_{n_1-1}
-{\bf X_1^{n_1-4}}x_{n_1-3}x_{n_1-1}
-{\bf X_1^{n_1-4}}x_{n_1-2}x_{n_1-1}
\right) \frac{1-\lf}{2}
+{\bf X_1^{n_1-1}}.
\end{align*}
This polynomial has degree $n-2$, which reduces the degree of the polynomial 
by $1$ and creates $2$ terms of the same type of degree $n-2$ and one of degree $n-3$.
Similarly, if we could modify the other degree $n-1$ polynomial in the same way 
it would bring down the overall degree of the real polynomial of the function to $n-2$.
This modification is central to our construction method and we now generalize it,
starting with defining two generic polynomial structures:
\begin{enumerate}
\item
$\hat{p}(i,s)= (\prod_{j=1}^i x_j)(1-x_{i+1})(1-x_{i+2}) (\prod_{j=s_i} x_j),
s \subset [n] \setminus [i+2]$,

\item
$
\tilde{p}(i,u)=(\prod_{j=1}^i x_j)(\prod_{j=u_i} x_j),
u \subset [n] \setminus [i].
$
\end{enumerate}
From these definitions we get the following two relations. 
\begin{equation}
\label{eq:t1}
\pdeg(\tilde{p}(i,u))=\pdeg(\hat{p}(i-2,u))= i+\abs{u},
\end{equation}
\begin{equation}
\label{eq:h1}
\tilde{p}(i,u)= \tilde{p}(i-k,u^k) \text{ where }u^k= u \cup \{i,i-1,\ldots i-k+1\}.
\end{equation}

Then corresponding to any $\tilde{p}(i,u)\frac{1+\lf}{2}$,
we have
\begin{align}
\label{eq:th1}
&\tilde{p}(i,u)\frac{1+\lf}{2} + \hat{p}(i-2,u)\frac{1-\lf}{2}
\\=&
\tilde{p}(i-2,u)\frac{1+\lf}{2}-
\tilde{p}(i-2,u)x_{i-1}\frac{1+\lf}{2}-
\tilde{p}(i-2,u)x_i\frac{1+\lf}{2}
+
\tilde{p}(i,u)
\nonumber \\=&
\tilde{p}(i-2,u)\frac{1+\lf}{2}-
\tilde{p}(i-2,u_1)\frac{1+\lf}{2}-
\tilde{p}(i-2,u_2)\frac{1+\lf}{2}
+
\tilde{p}(i,u),
\nonumber
\end{align}
where $\abs{u_1}=\abs{u_2}=\abs{u+1}$.
Then this addition of the polynomial 
$\hat{p}(i-2,u)\frac{1-\lf}{2}$
reduces the degree of the polynomial by $1$ 
and forms two similar terms of degree $1$ less,
one term of degree $2$ less, and one term that is only 
defined on $\xv$ and therefore has degree less than $n_1$.
Next we note another property of $\hat{p}$ before proceeding 
with the final proof. 

\begin{lemma}
\label{lemma:4}
Corresponding to the polynomial  $p_1$ of an MM function $f_1$
where $\phi_1$ is nonzero in only the points 
$\bm{1}_{n_1}$, $\bm{1}_{n_1-2}$, $\bm{0}_{n_1}^i$
such that $f_1 \in \M_n$ ($s(f,\bm{1}_n)=n$)
then adding of any arbitrary number of polynomial terms
$\hat{p}(i,u)(\frac{1 \pm \lf}{2})$ to $p_1$
transforms it to another MM Boolean function in $\M_n$,
provided that for any two added polynomial terms 
$\hat{p}(i_1,u_1)$ and $\hat{p}(i_2,u_2)$ we have $\abs{i_1-i_2} \geq 2$
and $i<n_1-2, ~\forall\,  \hat{p}(i,u)$. 
\end{lemma}
\begin{proof}
Consider any two 
$\hat{p}(i_1,u_1)$ and $\hat{p}(i_2,u_2)$ with $i_1 \leq i_2-2$.
Then, from Equation~\eqref{eq:4} we have
that adding  
$\hat{p}(i_1,u_1)(\frac{1 \pm \lf}{2})$
and 
$\hat{p}(i_2,u_2)(\frac{1 \pm \lf}{2}$)
to a polynomial corresponding to an MM function
only works (the resultant function is still in $\M_n$) 
if $\phi(\xv)=0$ for all $\xv$ such that
$\hat{p}(i_1,u_1)(\xv)=1$ or $\hat{p}(i_2,u_2)(\xv)=1$
and moreover
$\hat{p}(i_1,u_1)$ and $\hat{p}(i_1,u_2)$ should not both 
return $1$ for any input $\xv$.

The first condition is true as any polynomial of the type $\hat{p}(i,u)$
can only return $1$ if both $x_{i+1}$ and $x_{i+2}$ are set to $0$.
This cannot happen for any critical point as their weight is minimum $n_1-1$.
And since we have defined $i<n_1-2$ none of these polynomials can return $1$ 
for the point $\bm{1}_{n_1-2}00$. 

Also for any two polynomials $p_1=\hat{p}(i_1,u_1)$ and $p_2=\tilde{p}(i_2,u_2)$
with $i_1 \leq i_2-2$ 
they cannot both be $1$ for any input $\var{a} \in \F^{n_1}$ 
for the simple reason that for $p_1(\xv)$ to be $1$,  
$x_{i_1}$ and $x_{i_1+1}$ should be set as $1$  
and for $p_2(\xv)$ to be $1$ both $x_{i_1}$ and $x_{i_2}$
should be set to $0$, by definition. 
This completes the proof.
\end{proof}

Finally, we use the results and equations we have defined so far 
to develop a stepwise construction technique to obtain 
a function with $s(f)=n$ and $\pdeg(f)= n- \Theta(\log n)$.

\begin{theorem}
\label{thm:3}
There exists a  Boolean function $f \in \M_n$ 
with $\pdeg(f)=n -\Theta (\log n)$.
\end{theorem}
\begin{proof}

We define the function $f_1$ as in Equation~\eqref{eq:f1}
and the corresponding polynomial is 
$p_1=p_1'(\xv,\yv)+
\left({\bf X_1^{n_1-2}} - {\bf X_1^{n_1-2}}x_{n_1-1} - {\bf X_1^{n_1-2}}x_{n_1} \right)  (\frac{1+ \lf}{2})
+ 
{\bf X_1^{n_1}}$ where $\deg(p')=n_1+1$.

We start with $f_1$ and recursively reduce the polynomial degree 
by adding polynomials of the form $\hat{p}(i,u)$ where $i$ starts with $n_1-3$ and  decreases by $2$ with every new polynomial that we add.
The terms 
${\bf X_1^{n_1-2}}(\frac{1+ \lf}{2})$,
${\bf X_1^{n_1-2}}x_{n_1}(\frac{1+ \lf}{2})$
and
${\bf X_1^{n_1-2}}x_{n_1-1}(\frac{1+ \lf}{2})$
can all be expressed as polynomials $\tilde{p}(i,u_i)(\frac{1 \pm \lf}{2})$ for some $u_i$.
Here $i=n_1-2$.
We also know from Equation~\eqref{eq:t1} that 
any such polynomial can also be represented as 
some $\tilde{p}(i-k,u_i^k)(\frac{1 \pm \lf}{2})$.

Then corresponding to the term $(-1)^c \tilde{p}(i-k,u_i^k)(\frac{1 +(-1)^d \lf}{2})$
we add the polynomial $\hat{p}(i-k-2,u_i^k)(\frac{1 + (-1)^{c+d+1} \lf}{2})$ 
and this results in at most $3$ new polynomial terms of the form $\tilde{p}(i-k-1, {u_i^k}')$
of degree at least one less
(specifically one polynomial of degree $2$ less and two polynomials of degree $1$ less).
For each new polynomial we add, if the last polynomial that we added was $\hat{p}(i,u)(\frac{1 \pm \lf}{2})$, then the next polynomial is $\hat{p}(i-2,u')(\frac{1 \pm \lf}{2})$, 
where $u$ and $u'$ depend on the lower degree polynomials that form. 
Therefore for any defined $n_1$ we can at most add 
$\lfloor \frac{n_1-2}{2} \rfloor$ polynomials. 
This polynomial term addition is valid and the modified function is still in $\M_n$
as we have shown in Lemma~\ref{lemma:4}.

At the first step we have $3$ such polynomial terms to which we add
a $\hat{p}$ polynomial.
Then, at the $i$-th step, if we have $3^i$ polynomials of the form $\tilde{p}$
of degree $n-i$. Corresponding to each such, we 
add polynomials $\hat{p}(k,u)$ where $n_1-3^i \geq  k \geq n_1-3^{i+1}$
for each polynomial $\tilde{p}(k+2,u)$, and this would form 
$3^{i+1}$ polynomials of degree at most $n-(i+1)$. 
This step can at most be continued for $z$ iterations, 
where $\sum_{i=1}^z  2 \times 3^z \leq n_1-1$,
that is, $z \leq \log_3 (n_1-1)$.
If we fix $n_1=\frac{n}{2}$ this gives us a function 
in $\M_n$ 
with polynomial degree $n- \log_3(\frac{n}{2})= n- \Theta(\log n)$.
\end{proof}
It is easy to see that when we have not used some $k$ bits of $\xv$
as the zero bits in $\hat{p}$, our modifications do not affect the 
map $\phi$ in the point with at least any two of those $k$-bits being zero. 
This gives us the following corollary. 

\begin{corollary}
There are $\Omega \left(2^{2^k} \right)$ functions in $\M_n$ with polynomial degree 
$n- \log_3 (\frac{n}{2}) + \log_3(k)$ that we can recover using the 
construction method of Theorem~\textup{\ref{thm:3}}.
\end{corollary}

\section{Super-constant sensitivity with MM functions \label{mm:sup}}

\subsection{Proof of Lemma~\ref{lem:sym} \label{app:lem:sym}}
\begin{proof}
We have ${\sf sym}^k_m(\bm{0}_m)=0$ by definition. 
Now it suffices to see for any input point $\var{b} \in \F^m$  that can be obtained by flipping $i \leq k$ 
bits of $\bm{0}_m$ returns true for an odd number of monomials of ${\sf sym}^k_m$, in which case we will 
have ${\sf sym}^k_m(\var{b})=1$. This can be verified as follows.

Any input point obtained by flipping $i$ points in $\bm{0}_m$ has exactly $i$ variables set as one 
and the rest as zero. These $i$ values will return 1 for $i$ monomials of degree $1$, 
${i \choose 2}$ monomials of degree $2$, and so on. Thus the total number of monomials for which it will return $1$ is 
$\sum_{j=1}^i{i \choose j}= 2^i-1$, which is odd for all values of $i>0$. This completes the proof.

\end{proof}

\subsection{Proofs of Constructions~\ref{lem:korder:c} and \ref{th:f10:c} \label{mm:sup1} }

\begin{lemma}
\label{lem:korder}
Any function $f: \F^{n_1+n_2} \rightarrow \F$ with the  algebraic normal form
$\displaystyle f(\xv,\yv)=\bigoplus_{\var{a} \in \F^{n_1}} \left( ac_{\var{a}}(\xv) \cdot g_{\var{a}}(\yv) \right),$
where $g_{\bm{1}_{n_1}}= {\sf sym}^k_{n_2}$ 
and $g_{\var{a}}= 1$ for all $\var{a} \in \F^{n_1} : n_1 >wt(\var{a}) \geq n_1-k$
is $k$-th order sensitive.
\end{lemma}
\begin{proof}
We show that the function is $k$-th order sensitive in the point $(\var{a},\var{b})=(\bm{1}_{n_1},\bm{0}_{n_2})$.
We have $f(\bm{1}_{n_1},\bm{0}_{n_2})={\sf sym}^k_{n_2}(\bm{0}_{n_2})=0$.
Now, we consider any input that can be obtained by flipping at most $k$ bits of $(\bm{1}_{n_1},\bm{0}_{n_2})$.
There can be two cases.

{\bf All bits are flipped in $\yv$:}
In this case the input point is of the form 
$(\xv',\yv')=\left(\bm{1}_{n_1},\left[ \bm{0}_{n_2} \right]_k \right)$
for which the output is of the form 
$f(\xv',\yv')={\sf sym}^k_{n_2} \left(\left[ \bm{0}_{n_2} \right]_k \right)$
which is $1$ from Lemma~\ref{lem:sym}.

{\bf At least one bit is flipped in $\xv$:}
Any such point is of the form $(\xv',\yv')=
\left( \left[ \bm{1}_{n_1} \right]_k,\left[ \bm{0}_{n_2} \right]_{k-1} \right)$.
Then the output of the function is of the form 
$
f(\xv',\yv')
=g_{\left[ \bm{1}_{n_1} \right]_k}(\yv)
=1.
$

This shows that the output of the function is flipped if any $1 \leq i \leq k$ of the input bits are 
flipped at the point $(\bm{1}_{n_1},\bm{0}_{n_2})$.

\end{proof}


Finally we extend the technique of Section~\ref{log-sep} to obtain 
non-constant separation between number of variables and real polynomial 
degree in functions with super-constant order of sensitivity.

\begin{theorem}
\label{th:f10}
There exists a $k$-th order sensitive function in $\M_n^k$ with $n- \frac{\log \left( \frac{n}{2} -k \right)- \log{k}}{k}$
real polynomial degree.
\end{theorem}
\begin{proof}
We construct the corresponding function analogous to the result for first order order sensitive functions defined 
in Theorem~\ref{thm:3}.

We start with partial description of an MM type function $f:\F^{n_1+n_2}\to \F$ with nonlinear 
functions in $\F^{n_2}$ attached to the points in $\F^{n_1}$. 
As we have discussed, this function can be written as
$f(\xv,\yv)= \bigoplus_{\var{a} \in \F^{n_1}} \left( ac_{\var{a}}(\xv) \cdot g_{\var{a}}(\yv) \right)$,
 where $ g_{\var{a}}: \F^{n_2} \rightarrow \F, ~\forall \var{a},
$
and the corresponding real polynomial is written as 
$
p(\xv,\yv)= 
\displaystyle\sum_{\var{a} \in \F^{n_1}}
\pr_{\var{a}}(\xv) \hat{g}_{\var{a}}(\yv)
$
where $\hat{g}_{\var{a}}(\yv): \F^{n_2} \rightarrow \mathbb{R}$ is the real polynomial corresponding to $g_{\var{a}}(\yv)$.

The function $f$ has $g_{\var{a}}$ defined for all points with $wt(\var{a})\geq n-k$, which consists of 
$\sum_{i=0}^k {k \choose i}$ points, with $g_{\bm{1}_{n_2}}={\sf sym^n_k}$.
Let us denote by $\sym{n}{k}(\yv)$ the corresponding real polynomial.
Then the real polynomial corresponding to $\overline{\sf sym^n_k}$
is $1-\sym{n}{k}(\yv)$.

From this point our construction is analogous to that of Theorem~\ref{thm:3}.
In case of first order sensitivity, the polynomial terms we added to reduce the 
degree was of the form 
$\left( \prod_{i=1}^{j_1}(x_i)(1-x_{j_1+1})(1-x_{j_1+2})\prod_{j \in S}x_j \right) g(\yv)$, 
where $g(\yv) \in \{ \frac{1+\lf}{2}, \frac{1-\lf}{2} \}$ and $S \subseteq [n_1]-[j_1+2]$
so that high degree polynomial terms cancel out.  

In case of $k$-th order sensitive function, any new polynomial term that 
we add will be of form 
\begin{equation*}
T(j,S)=
\left( 
\prod_{i_1=1}^{j}(x_{i_1}) 
\prod_{i_2=j+1}^{j+1+k}\left(1-x_{i_2} \right)
\prod_{i_3 \in S} x_{i_3} \right) 
g(\yv), 
\end{equation*}
where $g(\yv) \in \{ \sym{n}{k}(\yv), 1-\sym{n}{k}(\yv) \}$ and $S \subseteq [n_1]-[j+k+1]$.

Let us suppose the total degree of the constructed polynomials at some point 
is $n'$. If we add a term of the form $T(j,S)$ that cancels out a 
polynomial term of degree $n'$, then it creates $k \choose {i-1}$ terms of degree 
$n'-i$ where $0<i<k$. Moreover with a new polynomial term that is added, 
the value of $j$ reduces by $k$ and as we know $j\geq 0$.
Then we have the following lower bound on how much can we reduce the polynomial 
degree of the function depending on the value of~$k$. 

Let us assume the polynomial degree of the function is $n-p$.
To reduce the polynomial degree by $p$ we have to cancel 
$(2^k-1)^p$ terms, and adding such a term will reduce the value of $j$ by 
$k$. Let us assume that $n_1=\frac{n}{2}$. 
Then we have 
\begin{align*}
&k \cdot (2^k-1)^p = \frac{n}{2}-k
& \implies &  \log(k) + p\log(2^k-1)= \log \left( \frac{n}{2} -k \right)
\\
 \implies &  \log(k) + pk \geq \log \left( \frac{n}{2} -k \right)
 &\implies &  p \geq \frac{\log \left( \frac{n}{2} -k \right)- \log(k)}{k}.
\end{align*}
Thus, even if we have $k=o(\log n)$ we get a non-constant separation 
between number of variables $n$ and the real polynomial degree~$n-p$.

\end{proof}

\end{document}